\documentclass[journal, 10.5pt]{IEEEtran}
\usepackage{amsmath}                
\usepackage{amssymb}  
\usepackage{psfrag,graphicx}
\usepackage{epsfig}
\usepackage{tikz}
\usepackage{multicol}
\usepackage{blindtext}        
\newtheorem{theorem}{\textbf{Theorem}}

\newtheorem{props}{\textbf{Proposition}}
\newtheorem{defin}{\textbf{Definition}}

\newtheorem{problem}{\textbf{Problem}}
\newtheorem{cor}{\textbf{Corollary}}

\def\R{\mathbb{R}}
\def\E{\mathbf{E}}

\def\F{\mathbf{F}}
\def\G{\mathbf{G}}
\def\H{\mathbf{H}}
\def\B{\mathbf{B}}

\def\N{\mathcal{N}}
\def\QED{~\rule[-1pt]{6pt}{6pt}\par\medskip}
\newenvironment{proof}{{\bf Proof.\ }}{ \hfill \QED}  

\title{\LARGE \bf Optimal Communication of States of Dynamical Systems over Gaussian Channels with Noisy Feedback: The Scalar Case}
\author{Ather Gattami\\
Ericsson Research\\
F\"ar\"ogatan 6, Stockholm, Sweden.\\ 
E-mail: ather.gattami@ericsson.com\\[8mm]
}

\begin{document}
\maketitle

\begin{abstract}
We consider the problem of communicating the state of a dynamical system via
a Shannon Gaussian channel. The receiver,
which acts as both a decoder and estimator, observes the noisy measurement of the channel
output and makes an optimal estimate of the state of the dynamical system in
the minimum mean square sense. Noisy feedback from the receiver to the transmitter is present. The transmitter observes the noise-corrupted feedback message from the receiver together with a possibly noisy measurement of the state the dynamical system. These measurements are then used to encode the message to be transmitted over a noisy Gaussian channel,
where a per symbol power constraint is imposed on the transmitted message.
Thus, we get a mixed problem of Shannon's source-channel coding problem and a sort of Kalman filtering problem. In particular, we consider two feedback instances, one being feedback of receiver measurements and the second being the receiver's state estimates.
We show that optimal encoders and decoders are linear filters with a finite memory and we give explicitly the state space realizations of the optimal filters. For the case where the transmitter has access to noisy measurements of the state, we derive a separation principle for the optimal communication scheme.  Furthermore, we investigate the presence of noiseless feedback or no feedback from the receiver to the transmitter. Necessary and sufficient conditions for the existence of a stationary solution are also given for the feedback cases considered.
\end{abstract}


\section*{Notation}
\begin{tabular}{ll}
$x^t$		& $x^t = (x(0), x(1), ..., x(t))$.\\
$\mathbb{L}$ & The set of lower triangular matrices.\\
$\B$  	& Denotes the backward shift operator,\\ 
		& $x(t-1) = \B x(t)$.\\			
$\mathbf{E}\{\cdot\}$ 	& $\mathbf{E} \{x\}$ denotes the expected value of the\\ 
			& stochastic variable $x$.\\
$\mathbf{E}\{\cdot|\cdot\}$ 	& $\mathbf{E} \{x|y\}$ denotes the expected value of the\\
			& stochastic variable $x$ given $y$.\\
$\mathbf{cov}$ & $\mathbf{cov}\{x,y\} = \E\{xy^\intercal \}$.\\
$h(x)$ 		& Denotes the entropy of $x$.\\
$h(x|y)$ 		& Denotes the entropy of $x$ given $y$.\\
$I(x;y)$ 		& Denotes the mutual information between\\
			& $x$ and $y$.\\

$\mathcal{N}(m,V)$  & Denotes the set of Gaussian variables with\\ 
				& mean $m$ and covariance $V$.\\
\end{tabular}


\begin{figure}
\label{general}
	\center
  	\includegraphics[width = 1.2\columnwidth]{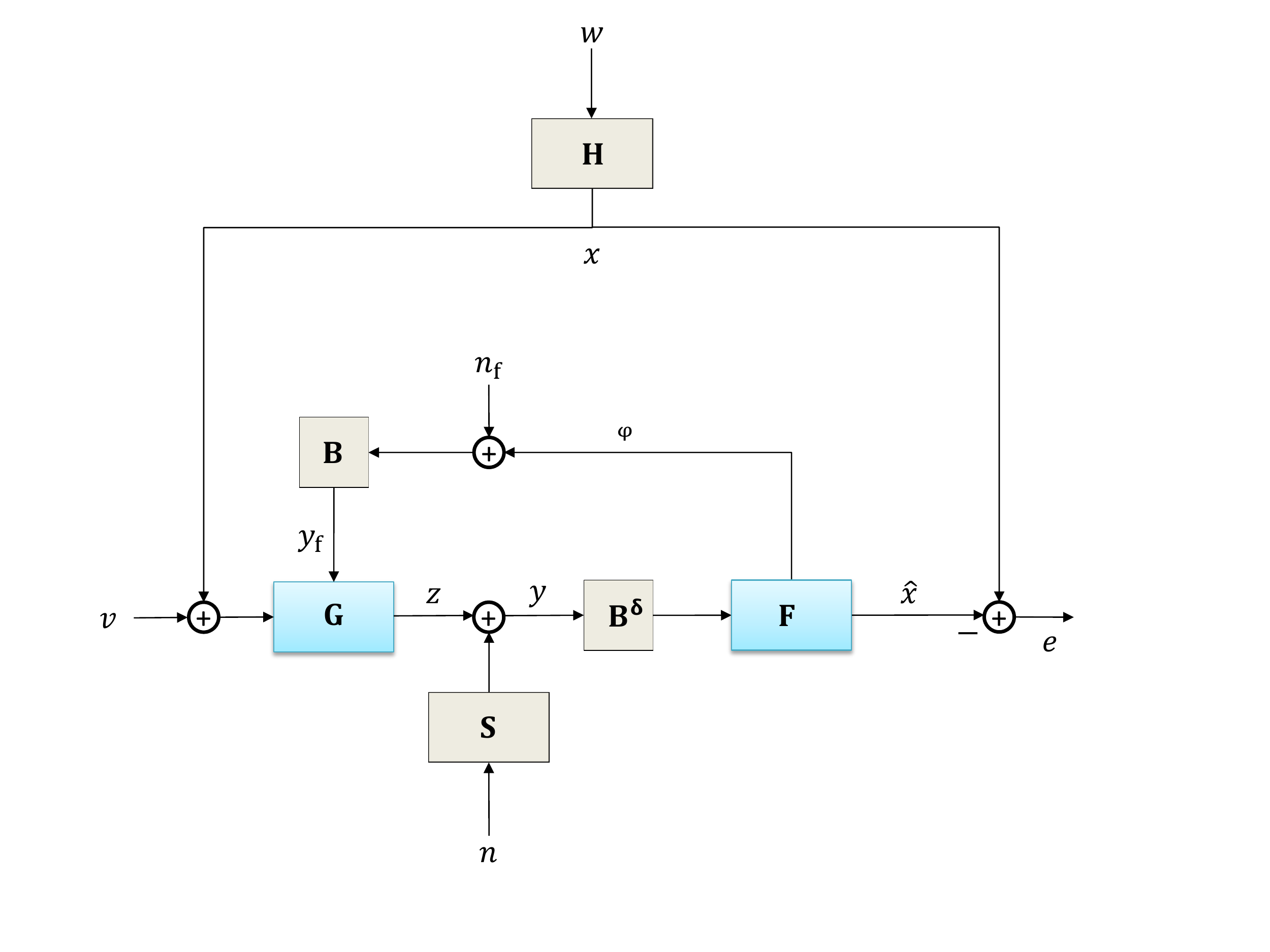}
	\caption{A simple model of an estimation problem of the state of the dynamical system $\H$ over a Gaussian communications channel with Gaussian noise $n\sim\mathcal{N}(0,N)$, Gaussian noise
	$n_\textup{f} \sim\mathcal{N}(0,N_\textup{f})$ for the feedback channel, the coloring filter $\mathbf{S}$ of the measurement noise $n$, and delay given by the backward shift operator 
	$\B$. The optimization parameters are given by the encoder $\G$ and the decoder $\F$. The symbols of the encoder output $z$ are power limited with $\E|z(t)|^2 \leq P$. }
\end{figure}

\section{Introduction}
\subsection{Background}
Many problems in practice require state estimation of a dynamical system where the possibly noisy state measurements at one end are transmitted over a noisy communciation another end where the state estmation is to be performed.

Shannon \cite{shannon:48, shannon1949} considered the problem of reliable communication of a one-dimensional source over a one-dimensional Gaussian channel. In particular, Shannon considered the  following coding-decoding setting for an analog Gaussian channel:
$$
\inf_{\substack{ f:\R\rightarrow \R\\ g:\R\rightarrow \R\\ \E |g(x)|^2 \leq P}}
\E|x - f(g(x)+n)|^2 
$$
where $x\sim \mathcal{N}(0,X)$, $n\sim \mathcal{N}(0,N)$, and $f, g$ are arbitrary functions with $\E |g(x)|^2 \leq P$. Shannon showed
that the infimum can be attained by using linear encoder and decoder $g$ and $f$, respectively.
The generalization of Shannon's result to higher dimensions is still open and there are examples where linear coding and decoding strategies might not be optimal \cite{pilc:1969}.

An important generalization of Shannon's AWGN channel is the case when the message $x$ to be estimated is the state of a given linear dynamical system driven by process noise. For instance, this problem arises in video-streaming over a wireless channel. A video stream consists of highly correlated information described by a dynamical system due to the correlation between the sequential picture frames. This is an instance of the general MIMO communcation problem with causality constraints, which adds structure to the problem. 
Another generalization is when the measurement noise is colored with the coloring filter given by a linear filter $\mathbf{S}$, see Figure \ref{general} for an illustration of the generalized communication system.


\begin{figure}
	\center
  	\includegraphics[width = 1\columnwidth]{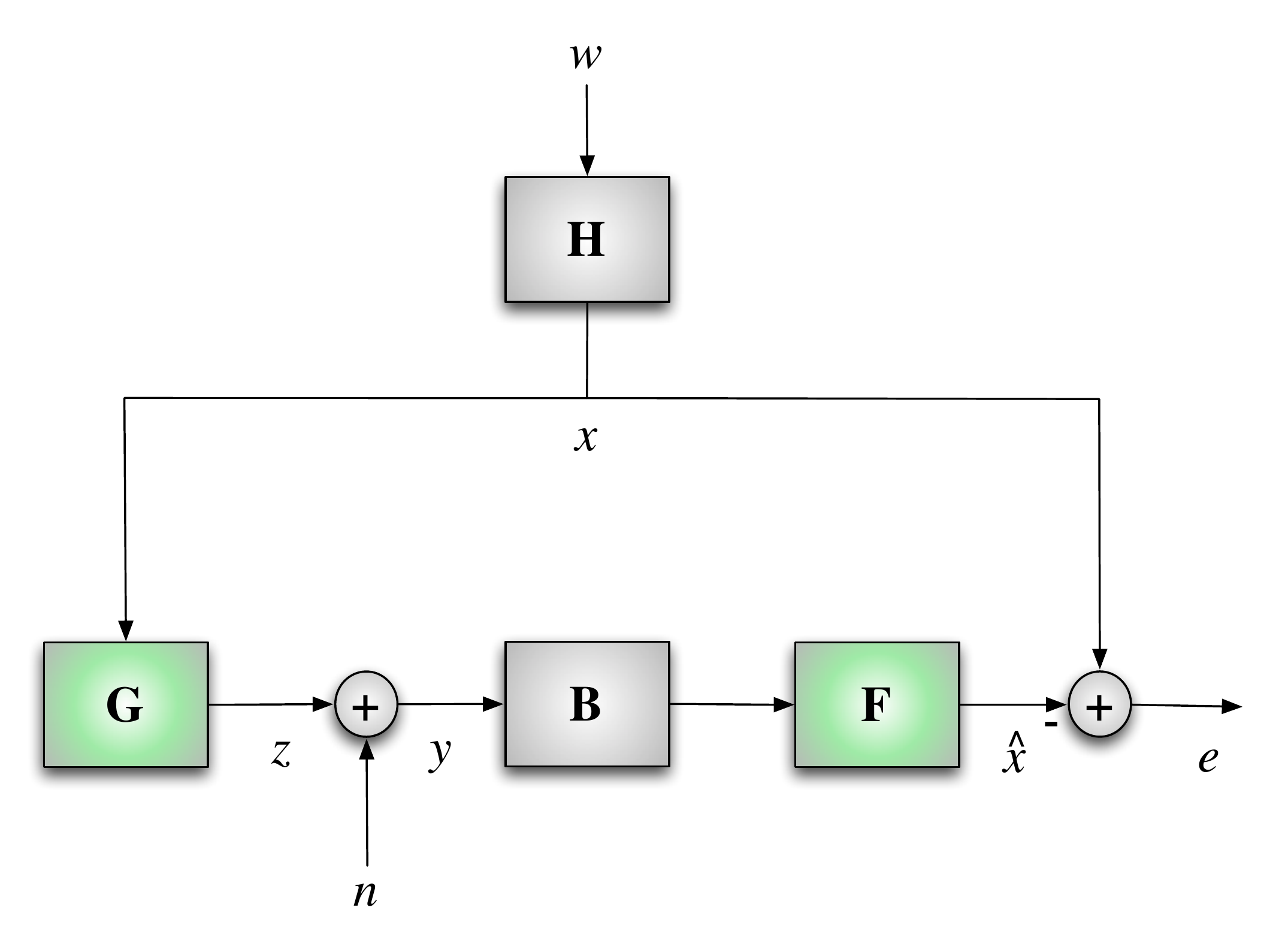}
	\caption{A simple model of an estimation problem of the state of the dynamical system $\H$ over a Gaussian communications channel with Gaussian noise $n\sim\mathcal{N}(0,N)$ and delay given by the backward shift operator 
	$\B$. The optimization parameters are given by the encoder $\G$ and the decoder $\F$. The samples of the encoder output $z$ are power limited with $\E|z(t)|^2 \leq P$. }
	\label{noFB}
\end{figure}


More specifically, consider the block-diagram in Fig. \ref{noFB}. We have the process noise given by $w$, which is assumed to be Gaussian white noise, and the state
is given by $x = \H w$ where $\H$ is a causal linear operator/filter. \\
The precoder is given by the causal operator
$\G$, not necessarily linear. The encoded signal $z = \G x$ is then transmitted over a Gaussian channel with white noise given by $n$. Typically, one has power constraints on the transmitted signal $z(t)$, that is $\E |z(t)|^2 \leq P$, for some positive real number $P$. At the other end, the message received is $y(t) = z(t)+n(t)$, for $t=0, ..., T-1$, and is delayed with $d$ time steps by the backward shift operator $\B$. Finally, the causal operator $\F$ is the decoder, designed to reconstruct the state $x$ by $\hat{x}=\F \B y$, to minimize the mean squared error $\E|e|^2 = \E |x-\hat{x}|^2$. 

For the case where $\G$ is a fixed linear operator, the optimal filter $\F$ is well known to be given by the optimal Kalman filter, which is a linear operator. However, if $\G$ is a precoder to be co-designed together with $\F$, we get a nonconvex problem even if we restrict the optimization problem to be carried out over linear operators/filters. To this date, it's not known if linear filters are optimal, and whether the order of the linear optimal filters is finite for the general MIMO case.


\begin{figure}
	\center
  	\includegraphics[width = 1\columnwidth]{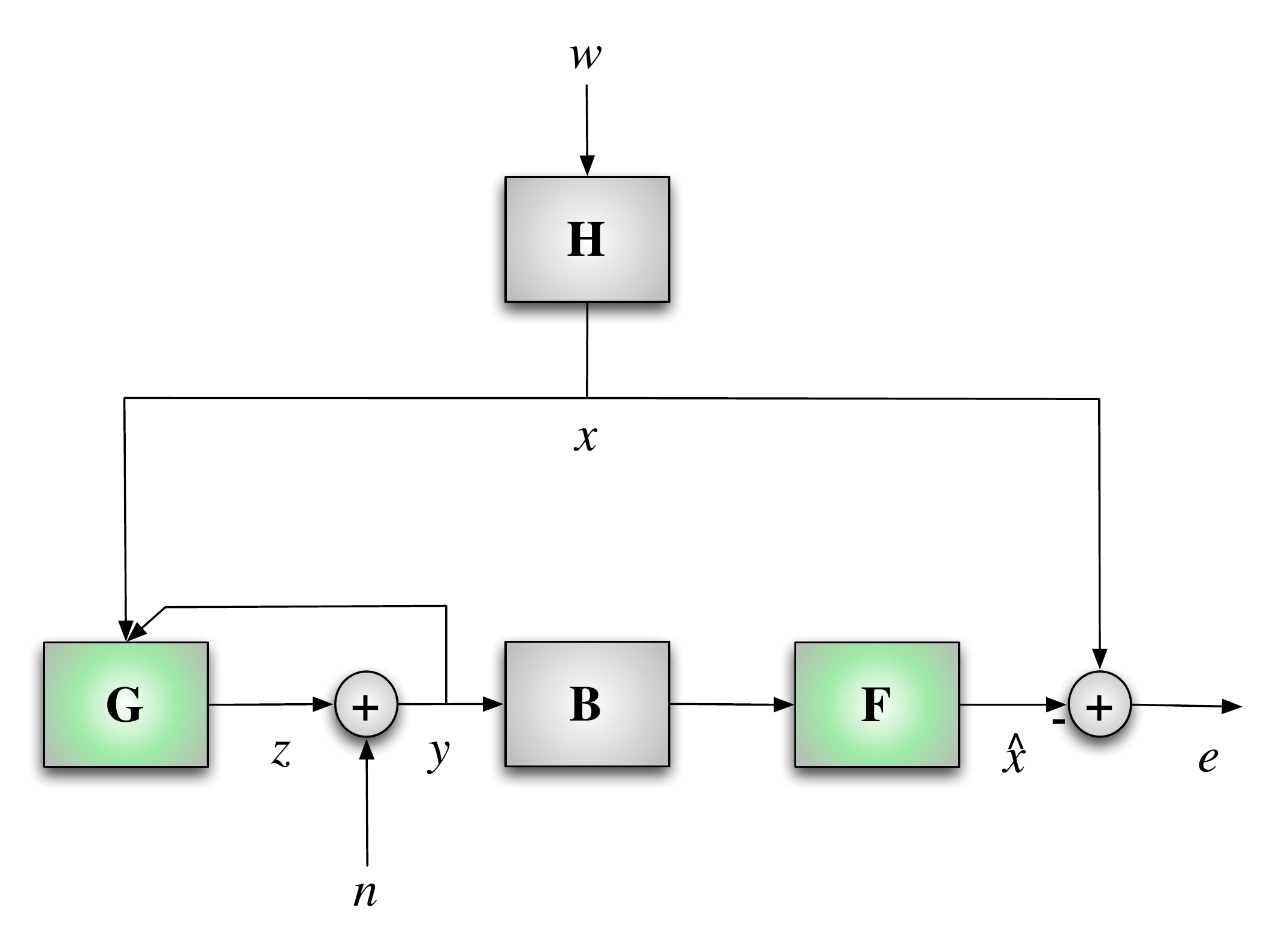}
	\caption{A simple model of a filtering problem over a Gaussian communications channel with noiseless feedback.}
	\label{noiselessFB}
\end{figure}

\subsection{Previous work}

Kalman \cite{kalman:1960} made a fundamental contribution to optimal control and filtering of linear dynamical systems by deriving recursive state space solutions. The model considered by Kalman assumes given linear measurements of the state, possibly partial and corrupted by noise. 
The solution relies on an orthogonality prinicple, where the filter update is based on an innovations process representing information that is orthogonal to the state estimate of the filter. 

The problem of optimal state estimation used for control of scalar dynamical systems was considered in \cite{bansal:basar:1988}, where noiseless feedback of the measurements at the receiver is present at the transmitter(see Figure \ref{noiselessFB}) and it was shown that linear filters where optimal. The role of a communication channel \textit{with feedback} and its effect on \textit{stability} was studied in \cite{tatikonda:2004} and necessary conditions for stability were given for linear time-invariant channels and that for time-varying channels was given in \cite{minero:2009}. Fundamental limitations of performance with sensitivity functions as a measure were studied in \cite{martins:2008}. The problem of communication and filtering over a noisy channel for the stationary case has been considered in \cite{joh10acc} where it was shown that this problem can be transformed to a convex optimization problem that grows with the size of the time horizon. However, the order of the  linear optimal filters obtained from \cite{joh10acc} is infinite. 

In another direction, \cite{kim:2010} studied the problem of source-channel coding over a communciation channel with colored noise with the correlation given by a linear filter $\mathbf{S}$, as depicted in Figure \ref{kim}. Here, the filter $\H$ is the identity(so $x=w$), $v=0$, and $\G$ encodes the information given by $w$ by using information of the measurements (with delay $d=1$) at the receiver through \textit{noiseless feedback}. Although the problem in \cite{kim:2010} considered maximizing the channel capacity, it was equivalent to the problem of minimizing the mean squared error of the state estimate as shown in Figure \ref{kim}. Also here, the solution relied on a sort of orthogonality principle where the transmitted information is orthogonal to that available at the receiver.

In \cite{gattami:ifac:2014}, preliminary results(with incomplete proofs) were given for the special case of communication and estimation without feedback for the scalar case as depicted in Figure \ref{noFB}. In all previous work, except \cite{bansal:basar:1988, joh10acc, gattami:ifac:2014}, average power constraints were assumed. Per symbol power constraints were considered in \cite{bansal:basar:1988, joh10acc, gattami:ifac:2014}.

\begin{figure}
	\center
  	\includegraphics[width = 1.2\columnwidth]{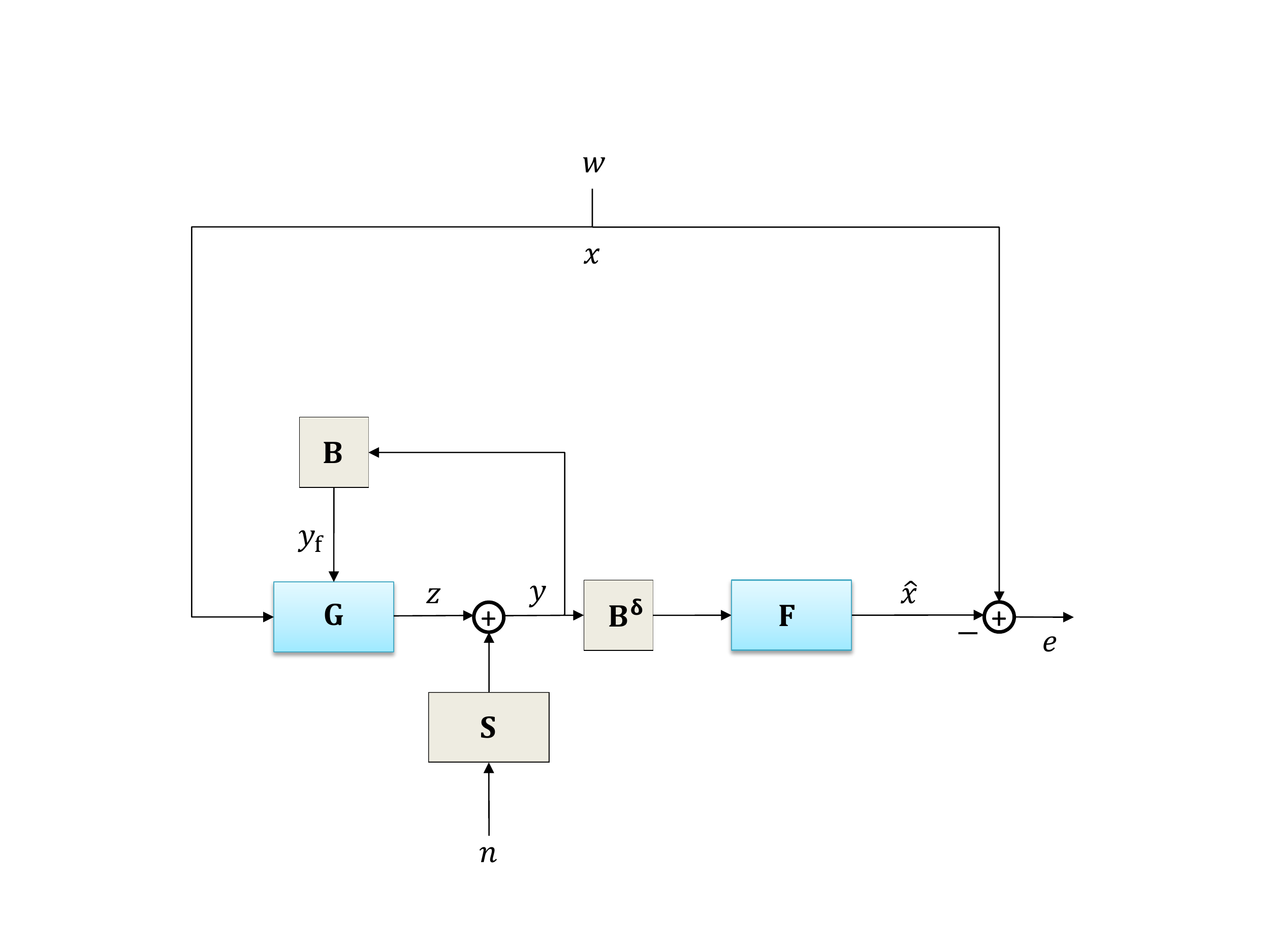}
	\caption{A simple model of a filtering problem over a Gaussian communications channel with noiseless feedback.}
	\label{kim}
\end{figure}

\subsection{Contributions}

We consider the linear dynamical system $\mathbf{H}$ given by
\begin{align} 
	&x(t+1) 	= ax(t) + bw(t)  \nonumber \\
	&x(0)=x_0, ~~~0\leq t\leq T-1. \nonumber
\end{align}

The main contributions of this paper is to derive the structure and explicit expressions of the optimal communication schemes as described in figures \ref{noS}, and \ref{yfb} and \ref{xhatfb} respectively, where \textit{noisy} feedback is present from the receiver side to the transmitter. We 
show that the optimal filters $\F$ and $\G$ are linear and have a finite memory independent of the size of the time horizon. In particular, we consider \textit{per symbol} power constraints on the transmitter signal as opposed to the average power constraints considered in the literature.
We show explicitly that the state space realizations of the optimal filters (for the case of full state measurement at the trasnmitter with delay at the receiver given by $\delta=1$) are given by

\begin{equation*}
 \G: \hspace{1cm} 
 \left\{\begin{aligned}
	s(t+1) 			&= a s(t) + K(t) (z(t) +  \hat{n}(t)) \\
	\hat{n}(t) 		&= \frac{N}{N + N_\textup{f}} (y_\textup{f}(t) - z(t))\\
	\check{x}(t) 	&=  x(t) - s(t)\\
 	z(t) 			&= \frac{\sqrt{P}}{\sigma_t} \check{x}(t),\\
 \end{aligned}\right.
 \end{equation*}
  
\begin{equation*}
 \F: \hspace{1.3cm} 
 \begin{aligned}
	\hat{x}(t+1) 	&= a\hat{x}(t) +  K(t) y(t)\\
 \end{aligned}
 \end{equation*}
 with $\E n^2(t) = N$, $\E n_\textup{f}^2(t) = N_\textup{f}$, 
$\sigma_t^2=  \E \check{x}^2(t)$,  
 $K(t) = a \sigma_t \sqrt{P}(P+N)^{-1}$, and $s(0) = 0$.

The interpretation of the state space equations is the following. 
$s(t) = \E\{\hat{x}(t) | x^{t}, y_\textup{f}^{t-1}\}$ is the estimate at the transmitter of the estimate $\hat{x}(t)$ at the decoder. The transmitter's estimate of $e(t)$ is 
$\check{x}(t) = \E\{e(t)  | x^{t}, y_\textup{f}^{t-1} \} = x(t) - s(t)$. This estimate is then transmitted over the Gaussian channel, in order to supply the decoder with the innovations(the incremental information the decoder needs to correct its estimate of $x(t)$).

We show that the error $e(t)$ may be stationary if and only if $|a|<1$. Then, we consider  the filtering problem over a communication channel, where noiseless feedback is introduced from the channel output to the precoder as depicted in Figure \ref{noiselessFB}. We show that the optimal transmitter and receiver are given by
  \begin{equation}
 \label{nofbeq}
 \begin{aligned}
	 \hat{x}(t) 	&= a \hat{x}(t) + K(t) y(t)\\
 	\tilde{x}(t) 	&= x(t)-\hat{x}(t)\\
 		z(t) 		&= \frac{\sqrt{P}}{\sigma_t} \tilde{x}(t),\\
 \end{aligned}
 \end{equation}
with 
 $$
 K(t) = a \frac{\sigma_t\sqrt{P}}{P+N},
 $$
and $\sigma_t^2=\E|\tilde{x}^2(t)|^2$ given by $\sigma_0^2 = \E x_0^2=V_{xx}(0)$ and
\begin{equation*}
  \sigma_t^2  = \frac{N}{N+P} \cdot a^2 \sigma_{t-1}^2 + b^2.
 \end{equation*}
Furthermore, we show that the error variance $e_t^2$ is bounded as $t\rightarrow \infty$ if and only if 
$$
\log_2(|a|) < C
$$
where $C$  is the capacity of the Gaussian channel from the transmitter to the receiver which is similar to previously published results in the context of stabilization of control system over communication channels \cite{tatikonda:2004}. We also consider the problem of communication under noisy feedback of the decoder's state estimates at the transmitter (see Figure \ref{xhatfb}). We find explicitly the optimum filter pair which is given by
\begin{equation*}
 \begin{aligned}
	\hat{x}(t+1) 		&= a\hat{x}(t) +  K(t) y(t)\\
	\check{x}(t+1)
			&= aN(P+N)^{-1} \check{x}(t) + x(t+1)-ax(t)  \\
			& ~~~+ a\bar{\sigma}_t^2 (\bar{\sigma}_t^2 + N_\textup{f})^{-1}(x(t) - \check{x}(t) - y_\textup{f}(t))\\ 		z(t) 			&= \frac{\sqrt{P}}{\sigma_t} \check{x}(t),\\
 \end{aligned}
 \end{equation*}
where $\sigma_t^2 =   \E \check{x}^2(t)$ and 
$\bar{\sigma}_t^2 = \E x^2(t) - \E\hat{x}^2(t) - \sigma_t^2$. We show that the estimation error is bounded as $t\rightarrow \infty$ if and only if the there exists a solution to the systems of nonlinear equations
\begin{equation*}
\begin{aligned}
		\sigma^2 = \frac{a^2 N^2}{(P + N)^{2}}\sigma^2 + \frac{a^2}{\bar{\sigma}^2 + N_\textup{f}}\bar{\sigma}^4 + b^2
\end{aligned}
\end{equation*}
and
\begin{equation*}
\begin{aligned}
		\bar{\sigma}^2 = \frac{a^2 N_\textup{f}^2}{\bar{\sigma}^2 + N_\textup{f}}\bar{\sigma}^2 + \frac{a^2PN}{(P + N)^2}\sigma^2 
\end{aligned}
\end{equation*}
The above equations are equivalent to a system of fourth order polynomial equations in two variables which can be solved efficiently using standard numerical tools.

\section{Preliminaries}

\begin{defin} 
The entropy of a real-valued stochastic variable $X$ with probability distribution $p(x)$ is defined as
$$
h(X) = -\int_{-\infty }^{\infty } p(x) \log_2 p(x) dx
$$
\end{defin}

\begin{defin}
For two real valued stochastic variables $X$ and $Y$, the conditional entropy of $X$ given $Y$ is
defined as
$$
h(X|Y) = h(X,Y) - h(Y).
$$
\end{defin}

\begin{defin}
The mutual information between $X$ and $Y$ is defined as
$$
 I(X,Y) = h(X) - h(X|Y) = h(Y) - h(Y|X).
$$
\end{defin}

\begin{props}[Entropy Power Inequality]
\label{entpineq}
	If $X$ and $Y$ are independent scalar random variables, then
	$$
	2^{2h(X+Y)} \geq 2^{2h(X)} + 2^{2h(Y)}
	$$
	with equality if $X$ and $Y$ are Gaussian stochastic variables. 
\end{props} 
\begin{proof}
	See \cite{cover:2006}, p. 674 - 675.
\end{proof}

\begin{defin}
	Random variables $X, Y, Z$ are said to form a Markov chain in that order if the conditional 
	distribution of $Z$ depends only on $Y$ and conditionally independent of $X$. This is denoted by
	$X\rightarrow Y \rightarrow Z$.
\end{defin}

\begin{props}[Data-Processing Inequality]
\label{dataproc}
	If $$X\rightarrow Y \rightarrow Z,$$ 
	then $$I(X;Z)\leq I(Y;Z).$$	
\end{props}
\begin{proof}
	See \cite{cover:2006}, p. 34-35.
\end{proof}

\begin{props}
\label{lse}
Let  $X$ and $Y$ be two stochastic variables. The optimal solution to the optimization problem
$$
\inf_{f(\cdot)} \E |X-f(Y)|^2
$$
is unique and given by the expectation of $X$ given $Y$
\[
\begin{aligned}
f_\star(Y) = \E\{X | Y\} .
\end{aligned}
\]
Furthermore, $f_\star(Y)$ and $X-f_\star(Y)$ are uncorrelated.  

\end{props}
\begin{proof}
Consult (\cite{shiryaev}, p. 237).
\end{proof}
\begin{props}
\label{logerror}
Consider the stochastic variables $X$ and $Y$, and let the estimation error of $X$ based on $Y$
be given by
$$
\tilde{X} = X - \E\{X|Y\}.
$$
Then,
\begin{equation}
	\label{lowerb}
		\frac{1}{2}\log_2{\det{(2\pi e \E \{\tilde{X}^2 \})}} \geq h(X|Y) = h(\tilde{X})
\end{equation}
with equality if and only if $X$ and $Y$ are jointly Gaussian. 
\end{props}  
\begin{proof}
Consult \cite{gamal:nit}, p. 21.
\end{proof}


\begin{figure}
	\center
  	\includegraphics[width = 1.2\columnwidth]{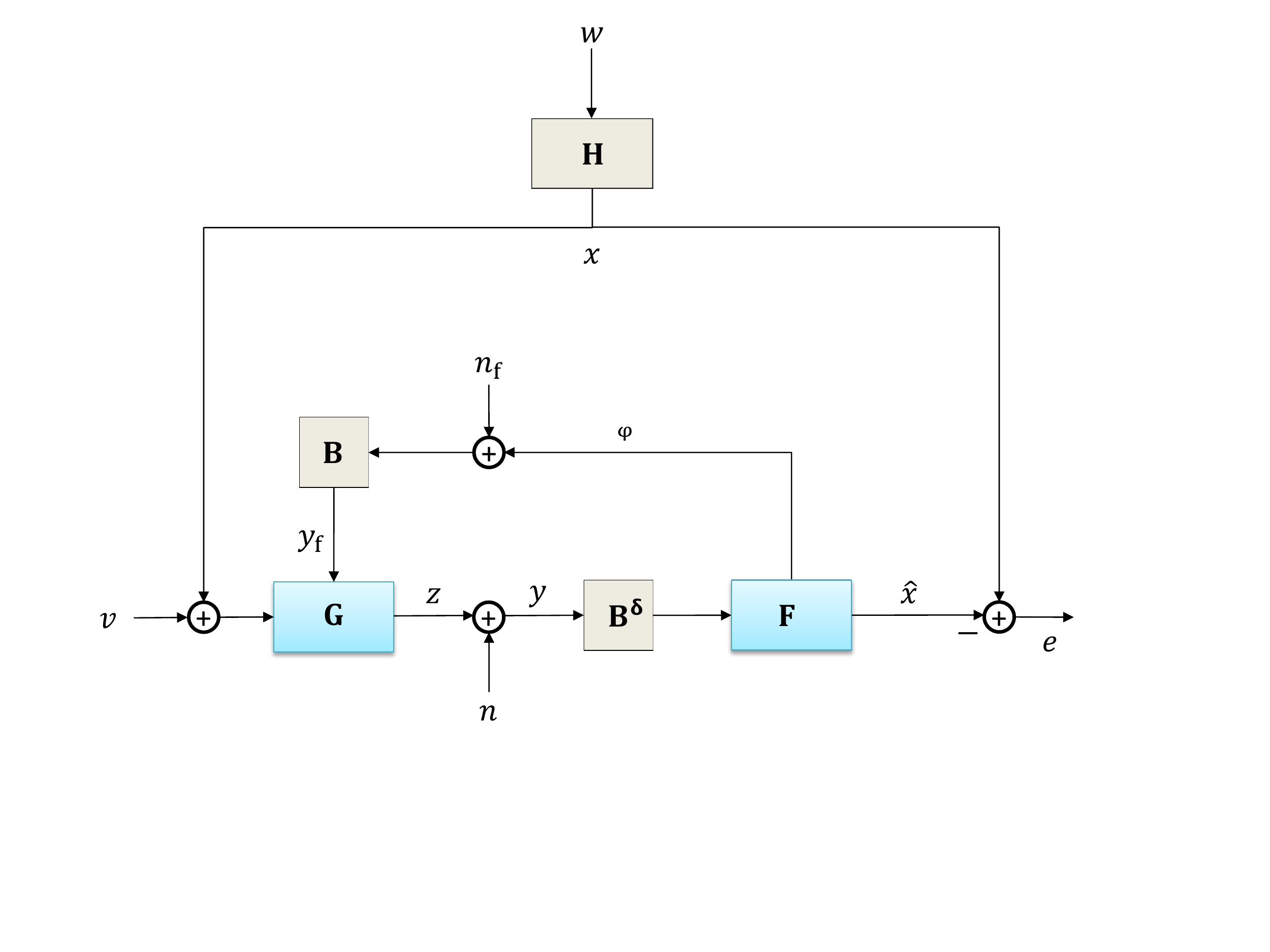}
	\caption{A simple model of an estimation problem of the state of the dynamical system $\H$ over a Gaussian communications channel with Gaussian noise $n\sim\mathcal{N}(0,N)$, Gaussian noise
	$n_\textup{f} \sim\mathcal{N}(0,N_\textup{f})$ for the feedback channel, and delay given by the backward shift operator 
	$\B$. The optimization parameters are given by the encoder $\G$ and the decoder $\F$. The samples of the encoder output $z$ are power limited with a peak power constraint given by $\E|z(t)|^2 \leq P$. }
	\label{noS}
\end{figure}

\begin{figure}
	\center
  	\includegraphics[width = 1.2\columnwidth]{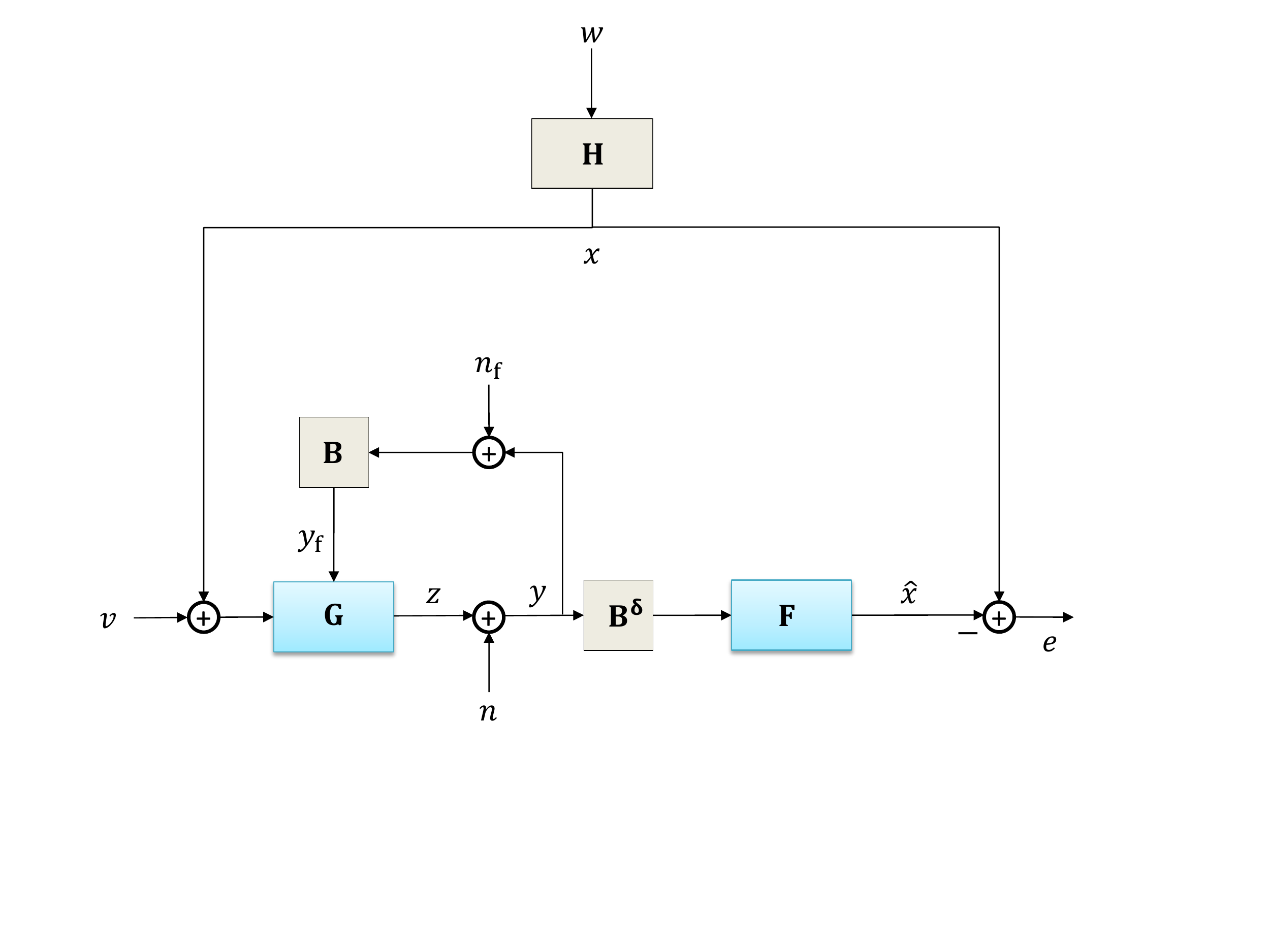}
	\caption{A simple model of an estimation problem of the state of the dynamical system $\H$ over a Gaussian communications channel with Gaussian noise $n\sim\mathcal{N}(0,N)$, Gaussian noise $n_\textup{f} \sim\mathcal{N}(0,N_\textup{f})$ for the feedback channel. We Here, we have feedback from the reciever side to the transmitter side in terms the reciever measurement $y(t)$.}
	\label{yfb}
\end{figure}

\begin{figure}
	\center
  	\includegraphics[width = 1.2\columnwidth]{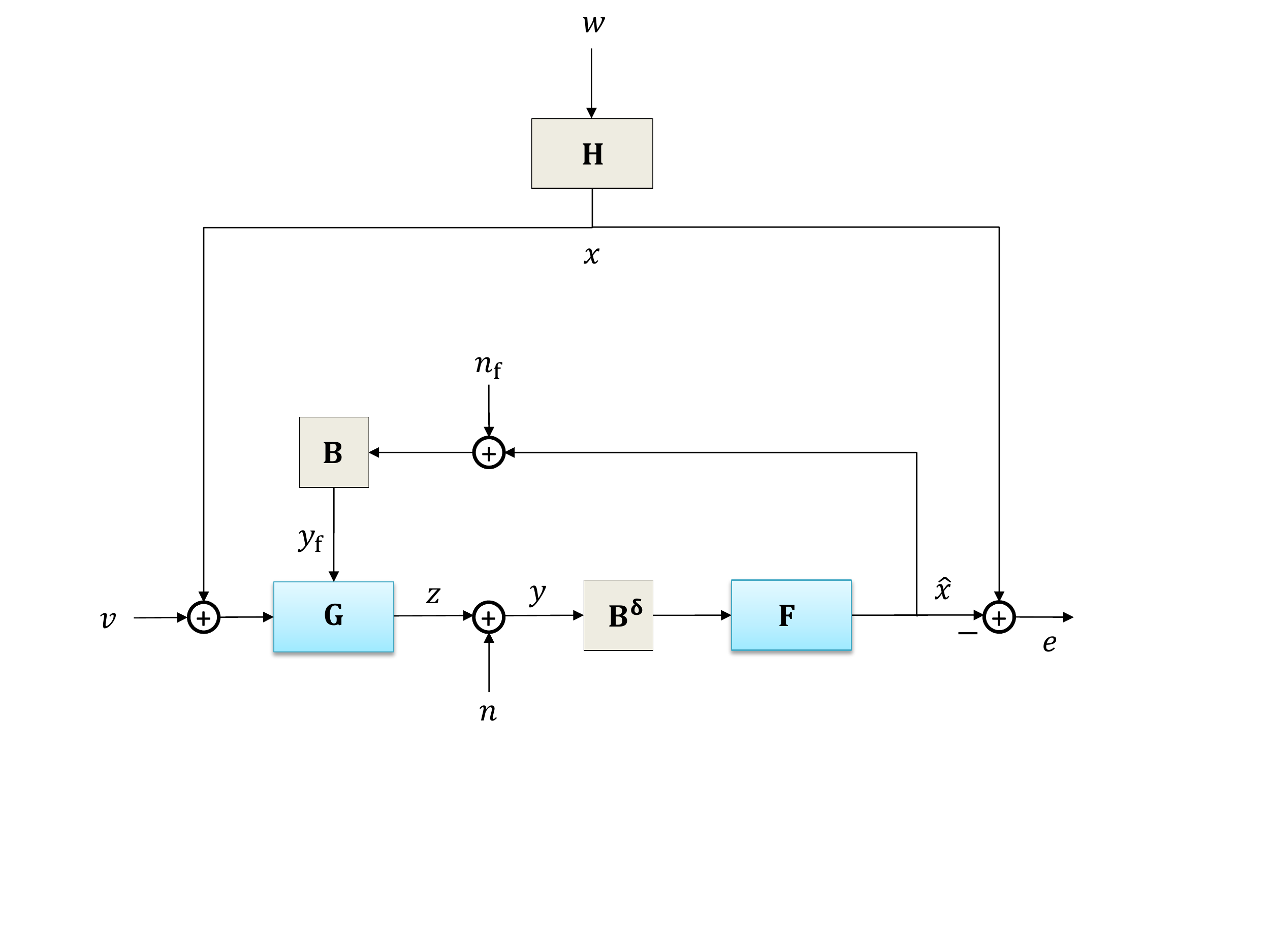}
	\caption{A simple model of an estimation problem of the state of the dynamical system $\H$ over a Gaussian communications channel with Gaussian noise $n\sim\mathcal{N}(0,N)$, Gaussian noise
	$n_\textup{f} \sim\mathcal{N}(0,N_\textup{f})$ for the feedback channel with feedback information given by the receiver's state estimates $\hat{x}(t)$.  }
	\label{xhatfb}
\end{figure}

\section{Problem Formulation}
We will consider the problem for the case $\textbf{S}=I$, as depicted in Figure \ref{noS}.

Let $\mathbf{H}$ be a first order linear time invariant dynamical system  with state-space realization

\begin{align} 
\label{linearsystem}
	x(t+1) 	&= ax(t) + bw(t) , ~~~~~~~ x(0)=x_0, ~~~0\leq t\leq T-1,
\end{align}
where $a,b \in\R$,  $\E x_0^2=V_{xx}(0)$,   
and $w$ is assumed to be white Gaussian noise with $w(t)\sim \N(0,1)$ for all $0\leq t\leq T-1$. 

The measurements at the decoder are given by $y(0) :=0$ and 
$$
y(t) = z(t)+n(t), ~~~~~ \text{for }t\geq 1,
$$
where $z$ is the transmitter signal and $n$ is a white Gaussian noise process with  $n(t)\sim \mathcal{N}(0, N)$. The decoder is a map given by $\F: y^{t-\delta}\mapsto \hat{x}(t)$. Without loss of generality, we will assume throughout that $\delta=1$ as the approach to the general case $\delta\ge 1$ is similar.

The transmitter receives the noisy feedback measurements 
$$
y_\textup{f}(t) = \phi(t)+n_\textup{f}(t), ~~~~~ \text{for }t\geq 1,
$$
where $n_\textup{f}$ is a white Gaussian noise process with  $n_\textup{f}(t)\sim \mathcal{N}(0, N_\textup{f})$. The encoder is a map given by $\G: (x^t, y_\textup{f}^{t-1}) \mapsto z(t)$.  We also have a per symbol power constraint on the transmitted signal $z(t)$ given by $\E |z(t)|^2 \leq P$.

The objective is to design causal precoder and decoder maps $\G: (x^t, y_\textup{f}^{t-1}) \mapsto z(t)$ and $\F: y^{t-1}\mapsto \hat{x}(t)$, respectively,  such that the average of the mean squared error 
$$
\frac{1}{T}\sum_{t=1}^T \E|x(t)-\hat{x}(t)|^2
$$
is minimized.
The precoder and decoder maps can be equivalently written as a causal dynamical system according to
  \begin{equation}
 \label{nonlin}
 \begin{aligned}
 z(t) 		&= g_t(x^t, z^{t-1}, y_\textup{f}^{t-1})\\
 y(t)		&= z(t)+n(t)\\
 y_\textup{f}(t)	&= \phi(t)+n_\textup{f}(t)\\
 \hat{x}(t) 	&= f_t(y^{t-1})
 \end{aligned}
 \end{equation}
where $g_t$ is the precoder and $f_t$ is the decoder.

\begin{problem}
\label{p1}
Consider the linear system
\begin{align*} 
	x(t+1) 	&= a(t)x(t) + b(t)w(t),
\end{align*}
$x(0)=x_0, ~0\leq t\leq T-1$, where $a(t),b(t) \in\R$, $\E x_0^2=V_{xx}(0)$, 
and $w$ is white Gaussian noise with $w(t)\sim \N(0,1)$, $0\leq t\leq T-1$.
Let $n$ and $n_\textup{f}$ be white Gaussian noise processes independent of each other and of $w$, with $n(t)\sim \mathcal{N}(0, N)$ and $n_\textup{f}(t)\sim \mathcal{N}(0, N_\textup{f})$. Find an optimal precoder and decoder pair (\ref{nonlin}) such that
$$
\frac{1}{T}\sum_{t=1}^T \E|x(t)-\hat{x}(t)|^2
$$
is minimized, where $y(0) = 0$.
 \end{problem}
 
Note that we haven't given the form of the function $\phi$. We will consider two cases of interest here, the first one being $\phi(t) = y(t)$, and the second one $\phi(t)=\hat{x}(t)$, as depicted in figures \ref{yfb} and \ref{xhatfb}, respectively.


 \section{Main Results}

\subsection{The Finite-Horizon Filtering problem with Receiver-Output Feedback}
The first result of this paper presents the structure of the optimal precoder and
decoder for the case where a noisy version of the receiver-output, $y(t)$, is available at the transmitter.

 \begin{theorem}
\label{mainy} 
Consider Problem \ref{p1} with $a(t) = a$, $b(t) = b$, and $\phi(t)=y(t)$. 
The optimal communication scheme is given by
  \begin{equation}
 \label{nofbeq}
 \begin{aligned}
	 \hat{x}(t) 	&= \E\{x(t) | y^{t-1}\}\\
 	\tilde{x}(t) 	&= x(t)-\hat{x}(t)\\
	\check{x}(t) &=  \E\{\tilde{x}(t) | x^{t}, y_\textup{f}^{t-1}\}\\
 		z(t) 		&= \frac{\sqrt{P}}{\sigma_t} \check{x}(t),\\
 \end{aligned}
 \end{equation}
 where $\sigma_t^2=  \E|\check{x}(t)|^2$, for $t= 1, ..., T$.
\end{theorem}
 \begin{proof}
See the Appendix.
\end{proof}
 
 \begin{theorem}
 \label{mainnofb2}
Consider Problem \ref{p1} with $a(t) = a$, $b(t) = b$, and $\phi(t)=y(t)$. 
The state space realization of the optimal communication scheme 
 is given by

\begin{equation}
 \label{nofbsseq}
 \begin{aligned}
	\hat{x}(t+1) 	&= a\hat{x}(t) +  K(t) y(t)\\
	s(t+1) 			&= a s(t) + K(t) (z(t) +  \hat{n}(t)) \\
	\hat{n}(t) 		&= \frac{N}{N + N_\textup{f}} (y_\textup{f}(t) - z(t))\\
	\check{x}(t) 	&=  x(t) - s(t)\\
 	z(t) 			&= \frac{\sqrt{P}}{\sigma_t} \check{x}(t),\\
 \end{aligned}
 \end{equation}
 where $s(0) = 0$, $V_{ss}(0) = V_{sx}(0) = 0$,

\begin{eqnarray}
\label{Kthm}
	&K(t) &= a \sigma_t\sqrt{P}(P+N)^{-1}\label{K} \\
	&\sigma_t^2 &= V_{xx}(t) -2V_{sx}(t)+ V_{ss}(t) \label{sigma}
\end{eqnarray}

\begin{equation}
\label{cov}
\begin{aligned}
&
\begin{bmatrix}
	V_{ss}(t+1) & V_{sx}(t+1)\\
	V_{xs}(t+1) & V_{xx}(t+1)	
\end{bmatrix} =  
\begin{bmatrix}
	\frac{K^2(t)N^2}{N+N_\textup{f}} & 0\\
	0 & b^2	
\end{bmatrix} 
 + \\
&\begin{bmatrix}
	\frac{aN}{P+N} 	& \frac{aP}{P+N}\\
	0		& a	
\end{bmatrix} 
\begin{bmatrix}
	V_{ss}(t) & V_{sx}(t)\\
	V_{xs}(t) & V_{xx}(t)	
\end{bmatrix}
\begin{bmatrix}
	\frac{aN}{P+N} 	& \frac{aP}{P+N}\\
	0		& a	
\end{bmatrix}^\intercal \\
\end{aligned}
\end{equation}

\end{theorem}
\begin{proof}
See the Appendix.
\end{proof}
 

\subsection{Time-Varying Systems}

The results considered so far treated the case where the state stems from a linear time invariant system. It's straight forward to verify that the results hold when we replace the parameters static $a, b, P, N, N_\textup{f}$ with time varying parameters $a(t), b(t), P(t), N(t), N_\textup{f}(t)$.


\subsection{Separation Principle for Optimal Communication}
Consider the linear system
\begin{align*} 
	x(t+1) 		&= ax(t) + bw(t) \\
	\gamma(t)	&= cx(k) + dv(t)
\end{align*}
for $0\leq t\leq T-1$, with $x(0)=x_0$, $\E x_0^2=V_{xx}(0)$, and $(w(t),v(t))$ is white Gaussian noise process with a given covariance. We assume now that the transmitter does't have
access to the state $x(t)$ but $\gamma(t)$ instead. We get the following problem.

\begin{problem}
\label{p2}
Consider the linear system
\begin{align*} 
	x(t+1) 	&= ax(t) + bw(t) \\
	\gamma(t)	&= cx(k) + dv(t)
\end{align*}
$0\leq t\leq T-$1, where $a,b \in\R$, $x(0)=x_0$,  $\E x_0^2=V_{xx}(0)$, and
$$
\E 
\begin{bmatrix}
	w(t)\\
	v(t)	
\end{bmatrix}
\begin{bmatrix}
	w(t)\\
	v(t)	
\end{bmatrix}^\intercal =
\begin{bmatrix}
	V_{ww}(t) & V_{wv}(t)\\
	V_{vw}(t) & V_{vv}(t)	
\end{bmatrix}
$$
is given for $0\leq t\leq T$.
Let $n$ and $n_\textup{f}$ be white Gaussian noise processes independent each other and of $w$, with $n(t)\sim \mathcal{N}(0, N)$ and $n_\textup{f}(t)\sim \mathcal{N}(0, N_\textup{f})$. Find an optimal precoder and decoder pair 
\begin{equation}
 \label{nonlin2}
 \begin{aligned}
 z(t) 		&= g_t(\gamma^t, z^{t-1}, y_\textup{f}^{t-1})\\
 y(t)		&= z(t)+n(t)\\
 y_\textup{f}(t)	&= \phi(t)+n_\textup{f}(t)\\
 \hat{x}(t) 	&= f_t(y^{t-1})
 \end{aligned}
 \end{equation}
 such that
$$
\frac{1}{T}\sum_{t=1}^T \E|x(t)-\hat{x}(t)|^2
$$
is minimized, where $y(0) = 0$.
 \end{problem}

The optimal transmission scheme is for the transmitter to find the best estimate of $x(t)$ based on $\gamma^t$, namely $\breve{x}(t) = \E\{x(t) | \gamma^t\}$, and then use this estimate as the state to be transmitted using the optimal communication scheme for the case of full state measurement at the transmitter given by (\ref{nofbsseq}). 

 \begin{theorem}
 \label{mainfb2}
The state space realization of the optimal communication scheme solution of 
Problem \ref{p2} with $\phi(t)=y(t)$ is given by

\begin{equation}
 \label{nmfbsseq}
 \begin{aligned}
	\hat{x}(t+1) 	&= a\hat{x}(t) +  K(t) y(t)\\
	s(t+1) 			&= a s(t) + K(t) (z(t) +  \hat{n}(t)) \\
	\hat{n}(t) 		&= \frac{N}{N + N_\textup{f}} (y_\textup{f}(t) - z(t))\\
	\check{x}(t) 	&=  x(t) - s(t)\\
 	z(t) 			&= \frac{\sqrt{P}}{\sigma_t} \check{x}(t),\\
 \end{aligned}
 \end{equation}
 where $s(0) = 0$, $V_{ss}(0) = V_{sx}(0)$, $V_{\xi \xi}(0) = V_{xx}(0)$,

\begin{equation*} 
\begin{aligned}
L(t)					&= V_{\xi \xi}(t)c(c^2V_{\xi \xi}(t) + d^2V_{vv}(t) )^{-1} \\ 
V_{\xi \xi}(t+1) 	&=  (a-aL(t)c)^2 V_{\xi \xi}(t) \\
					&+
\begin{bmatrix}
	b & -aL(t)
\end{bmatrix}
\begin{bmatrix}
	V_{ww}(t) & V_{wv}(t)\\
	V_{vw}(t) & V_{vv}(t)	
\end{bmatrix}
\begin{bmatrix}
	b & -aL(t)
\end{bmatrix}^\intercal
\end{aligned}
\end{equation*}

\begin{equation*} 
	\begin{aligned}
	\beta^2(t) &= L^2(t+1)  (c^2 V_{\xi \xi}(t+1)  + d^2 V_{vv}(t+1) )
	\end{aligned}
\end{equation*}

\begin{eqnarray}
\label{Kthm2}
	&K(t) &= a \sigma_t\sqrt{P}(P+N)^{-1}\label{K2} \\
	&\sigma_t^2 &= V_{xx}(t) -2V_{sx}(t)+ V_{ss}(t) \label{sigma2}
\end{eqnarray}

\begin{equation}
\label{cov2}
\begin{aligned}
&
\begin{bmatrix}
	V_{ss}(t+1) & V_{sx}(t+1)\\
	V_{xs}(t+1) & V_{xx}(t+1)	
\end{bmatrix} =  
\begin{bmatrix}
	\frac{K^2(t)N^2}{N+N_\textup{f}} & 0\\
	0 & \beta^2(t)	
\end{bmatrix} + \\
&\begin{bmatrix}
	\frac{aN}{P+N} 	& \frac{aP}{P+N}\\
	0		& a	
\end{bmatrix} 
\begin{bmatrix}
	V_{ss}(t) & V_{sx}(t)\\
	V_{xs}(t) & V_{xx}(t)	
\end{bmatrix}
\begin{bmatrix}
	\frac{aN}{P+N} 	& \frac{aP}{P+N}\\
	0		& a	
\end{bmatrix}^\intercal \\
\end{aligned}
\end{equation} 
\end{theorem}
\begin{proof}
The proof is deferred to the appendix.
\end{proof}


\subsection{No Feedback}
A special case is when no feedback is available from the receiver to the transmitter. This is equivalent to letting $N_\textup{f} \rightarrow \infty$, or setting $y_\textup{f} = 0$, as depicted in Figure \ref{noFB}. This will simply imply that $\hat{n} = 0$ and $\tilde{n} = n$, and thus, we obtain the optimal communication scheme that was previously obtained in \cite{gattami:ifac:2014}. The case of no feedback is very delicate, since it does \textit{not} possess the property of communicating information that is orthogonal to the information available at the receiver. 

\begin{cor}
 \label{mainnofb3}
The state space realization of the optimal communication scheme solution of 
Problem \ref{p1} with $y_\textup{f} = 0$ is given by

\begin{equation}
 \label{nofbsseq}
 \begin{aligned}
	\hat{x}(t+1) 	&= a\hat{x}(t) +  K(t) y(t)\\
	s(t+1) 			&= as(t) + K(t) z(t)\\
	\check{x}(t) 	&=  x(t) - s(t)\\
 	z(t) 			&= \frac{\sqrt{P}}{\sigma_t} \check{x}(t),\\
 \end{aligned}
 \end{equation}
 where $s(0) = 0$, $V_{ss}(0) = V_{sx}(0) = 0$,

\begin{eqnarray}
	&K(t) &= a \sigma_t\sqrt{P}(P+N)^{-1}\label{K} \\
	&\sigma_t^2 &= V_{xx}(t) -2V_{sx}(t)+ V_{ss}(t) \label{sigma}
\end{eqnarray}

{\small
\begin{equation}
\label{cov}
\begin{aligned}
&
\begin{bmatrix}
	V_{ss}(t+1) & V_{sx}(t+1)\\
	V_{xs}(t+1) & V_{xx}(t+1)	
\end{bmatrix} =\\
&
\begin{bmatrix}
	\frac{aN}{P+N} 	& \frac{aP}{P+N}\\
	0		& a	
\end{bmatrix} 
\begin{bmatrix}
	V_{ss}(t) & V_{sx}(t)\\
	V_{xs}(t) & V_{xx}(t)	
\end{bmatrix}
\begin{bmatrix}
	\frac{aN}{P+N} 	& \frac{aP}{P+N}\\
	0		& a	
\end{bmatrix}^\intercal \\
\end{aligned}
\end{equation} 
} 
\end{cor}


\subsection{Noiseless Feedback}

Another interesting special case, which has been solved in \cite{bansal:basar:1988}, is when we have perfect feedback from the receiver to the transmitter, as depicted in Figure \ref{noiselessFB}. We will reproduce this result using our approach, and furthermore, give necessary and sufficient conditons for the estimation error to be bounded for the case $|a|\ge 1$. 

Let
\begin{equation}
 \label{nobfxhat}
 \begin{aligned}
 \hat{x}(t|t) 	:&= \E\{x(t) | y^{t}\}\\
 			&= \E\{x(t) | y^{t-1}, y(t)\}\\
			&= \E\{\hat{x}(t)+\tilde{x}(t) | y^{t-1}, z(t)+n(t)\}\\
 			&= \hat{x}(t) + \E\left\{\tilde{x}(t) \Big{|}  \frac{\sqrt{P}}{\sigma_t}\tilde{x}(t) + n(t)\right\}\\
			&= \hat{x}(t) + \frac{\sigma_t\sqrt{P}}{P+N}\left( \frac{\sqrt{P}}{\sigma_t}\tilde{x}(t) + n(t) \right ) \\
 			&=  \hat{x}(t) + \frac{P}{P+N} \tilde{x}(t) + \frac{\sigma_t\sqrt{P}}{P+N}n(t),\\
			&
 \end{aligned}
 \end{equation}
and
 \begin{equation}
 \label{Xtilde}
 \begin{aligned}
 \tilde{x}(t|t) 	:&= x(t) - \hat{x}(t|t)\\
 			&= \hat{x}(t) + \tilde{x}(t) - \hat{x}(t|t)\\
 			&= \frac{N}{P+N}\tilde{x}(t) - \frac{\sigma_t\sqrt{P}}{P+N}n(t)
 \end{aligned}
 \end{equation}

 Also, (\ref{nobfxhat})-(\ref{Xtilde}) give

 \begin{equation}
 \label{Xhat2}
 \begin{aligned}
 \hat{x}(t+1) 	&= \E\{x(t+1) | y^t\}\\
 			&= \E\{ax(t)+ bw(t) | y^t\}\\
 			&= a\hat{x}(t|t),
 \end{aligned}
 \end{equation}
  and
  \begin{equation}
 \label{xtilde2}
 \begin{aligned}
 \tilde{x}(t+1) 	&= x(t+1) - \hat{x}(t+1)\\
 				&= a\tilde{x}(t|t) +bw(t)\\
				&= \frac{N}{P+N}a\tilde{x}(t) - \frac{\sigma_t\sqrt{P}}{P+N}an(t)+bw(t)
 \end{aligned}
 \end{equation}
 By considering the state estimation error dynamics in (\ref{xtilde2}), the reader might be tempted to conclude that the decoder will be able to track the state $x(t)$ if and only if
   $$
   \frac{N}{P+N} |a| <1.
   $$
 
 However, this conclusion is erroneous since the gain of the noise $n(t)$ depends on $\sigma_t = \sqrt{\E\{\tilde{x}^2(t)\}}$. What we need to consider is the dynamics of the \textit{variance} of the estimation error $\tilde{x}^2(t)$ as follows.
 
  \begin{equation}
 \label{x2tilde}
 \begin{aligned}
 \E\{\tilde{x}^2(t|t)\} 	&= \E\left\{\left(\frac{N}{P+N}\tilde{x}(t) - \frac{\sigma_t\sqrt{P}}{P+N}n(t)\right)^2\right\}\\
 					&= \left(\frac{N}{P+N}\right)^2  \E\{\tilde{x}^2(t)\}
					+ \left(\frac{\sigma_t\sqrt{P}}{P+N}\right)^2\E\{n^2(t)\}\\
					&= \frac{N^2}{(P+N)^2}\sigma_t^2
					+ \frac{\sigma_t^2 P}{(P+N)^2} N\\
					&= \frac{N}{P+N}\sigma_t^2\\
					&= \frac{N}{P+N}\E\{\tilde{x}^2(t)\}
 \end{aligned}
 \end{equation}
 Equations (\ref{xtilde2}) and (\ref{x2tilde}) give
 
 \begin{equation}
 \label{finalineq}
 \begin{aligned}
 \E \{\tilde{x}^2(t+1)\} 	&= a^2\E \{\tilde{x}^2(t|t)\} + b^2\E \{w^2(t)\}\\
 					&= \frac{N}{N+P}  a^2\cdot \E \{\tilde{x}^2(t)\} + b^2.
  \end{aligned}
 \end{equation}

The recurrence equation (\ref{finalineq}) implies that a stationary solution to Problem 1 for the case 
$n_\textup{f} = 0$ exists if and only if
$$
 1> \frac{{N}}{P+N}|a|^2,
 $$
which is equivalent to 
$$
\log_2(|a|) < \frac{{1}}{2}\log_2 \left(1+ \frac{P}{N}\right)
$$
Note that the capacity $C$ of the Gaussian channel is given by
$$
C = \frac{{1}}{2}\log_2 \left(1+ \frac{P}{N}\right)
$$
so a necessary and sufficient condition for the mean squared estimation error to be finite is 
$$
\log_2(|a|) < C.
$$
A similar result for stabilization of a control system over a discrete memoryless channel has been obtained in \cite{tatikonda:2004}.


\subsection{Stationarity}
In this section, we will present conditions under which a stationary solution exists to Problem \ref{p1} for the case $\phi(t) = y(t)$
(that is a solution as $T\rightarrow \infty$).
Let $\tilde{x}(t) = x(t) - \hat{x}(t)$ be the estimation error of $x(t)$ and consider the state space equations (\ref{nofbsseq}) of the optimal estimate. After some algebra, we get the state space equations for the estimation error (see (\ref{xtilde}) in the proof of Theorem \ref{mainnofb2} in the Appendix):
\begin{align*}
\tilde{x}(t+1) 	&= a\tilde{x}(t) - a\kappa(t) y(t) + bw(t)\\
				&= a\tilde{x}(t) - a\kappa(t) \frac{\sqrt{P}}{\sigma_t} \check{x}(t) - 
							a\kappa(t) n(t) + bw(t)\\
				&= a\tilde{x}(t) - \frac{aP}{P+N} \check{x}(t) - 
							K(t) n(t) + bw(t)\\			
				&= a\tilde{x}(t) - \frac{aP}{P+N} (\tilde{x}(t) - \bar{x}(t)) - 
						K(t) n(t) + bw(t)\\
				&= \frac{aN}{P+N} \tilde{x}(t) - K(t) n(t) + bw(t) + \frac{aP}{P+N}\bar{x}(t)
\end{align*}
 with
 $$
 \bar{x}(t+1) 	= a\bar{x}(t)- K(t) \tilde{n}(t).
 $$
 Obviously, for $\tilde{n}(t) \neq 0$(that is $N_\textup{f} >0$),  the state $\bar{x}(t)$ can be stationary if 
 and only if $|a|<1$. In addition, in order for $\tilde{x}(t)$ to be stationary, we must have
 $$
 1 > \frac{{N}}{P+N}|a|^2.
 $$
Clearly, the inequality above is always fulfilled for $|a|<1$. We conclude the result above:
 
 \begin{theorem} 
 	Problem \ref{p1} with $\phi(t) = y(t)$ has a stationary solution for $N_\textup{f} >0$ as $T\rightarrow \infty$ if and only if $|a|<1$ and there are no filters $\F$ and $\G$ that achieve a finite mean square error for $|a|\ge 1$. \\
\end{theorem}

It's interesting to see the difference between the noiseless feedback case and the noisy feedback one. This raises the question of whether the feedback function $\phi$ could be chosen differently in order to get filters that can track a state as the time horizon goes to infintiy. Indeed, this turns out to be the case as will be shown in the sequel.
Noiseless feedback of the output, $\phi(t)=y(t)$, makes the state estimates at the receiver available to the transmitter. This would equivalently correspond to the case of noiseless feedback of the state estimates, that is for $n_\textup{f}=0$ and $\phi(t)=\hat{x}(t)$ as shown in Figure \ref{xhatfb}.


\subsection{Noisy Feedback of the State Estimates}
Suppose that the receiver transmits its state estimates $\hat{x}(t)$ back to the transmitter overa noisy channel. Being inspired by the noiseless feedback results, we can construct the 
measurements $x(t) - y_\textup{f}(t) = \tilde{x}(t) - n_\textup{f}(t)$ and set 
$\check{x}(t) = \tilde{x}(t) - n_\textup{f}(t)$. However, this strategy is not necessarily optimal. The next result gives the optimal communication scheme.

 \begin{theorem}
 \label{mainxhat2}
Consider Problem \ref{p1} with $a(t) = a$, $b(t) = b$, and $\phi(t)=\hat{x}(t)$. 
The state space realization of the optimal communication scheme 
 is given by

\begin{equation}
 \begin{aligned}
	\hat{x}(t+1) 		&= a\hat{x}(t) +  K(t) y(t)\\
	\check{x}(t+1)
			&= aN(P+N)^{-1} \check{x}(t) + x(t+1)-ax(t)  \\
			& ~~~+ a\bar{\sigma}_t^2 (\bar{\sigma}_t^2 + N_\textup{f})^{-1}(x(t) - \check{x}(t) - y_\textup{f}(t))\\ 		z(t) 			&= \frac{\sqrt{P}}{\sigma_t} \check{x}(t),\\
 \end{aligned}
 \end{equation}
 where 
\begin{eqnarray}
	&K(t) &= a\sigma_t \sqrt{P}(P+N)^{-1}\label{K} \\
	&\sigma_t^2 &= \E \check{x}^2(t) 
\end{eqnarray}
$\sigma_0^2 = 0$, $\bar{\sigma}_0^2 = 0$,
\begin{equation}
\label{sigmanonlin}
\begin{aligned}
		\sigma_{t+1}^2 = \frac{a^2 N^2}{(P + N)^{2}}\sigma_t^2 + \frac{a^2}{\bar{\sigma}_t^2 + N_\textup{f}}\bar{\sigma}_t^4 + b^2,
\end{aligned}
\end{equation}
and
\begin{equation}
\label{sigmabar}
\begin{aligned}
		\bar{\sigma}_{t+1 }^2 = \frac{a^2 N_\textup{f}^2}{\bar{\sigma}_t^2 + N_\textup{f}}\bar{\sigma}_t^2 + \frac{a^2PN}{(P + N)^2}\sigma_t^2 
\end{aligned}
\end{equation}
\end{theorem}
\begin{proof}
See the Appendix.
\end{proof}

Now suppose that there is a stationary solution to (\ref{sigmanonlin}) - (\ref{sigmabar}). Then,
\begin{equation}
\label{sigmanonlinstat}
\begin{aligned}
		\sigma^2 = \frac{a^2 N^2}{(P + N)^{2}}\sigma^2 + \frac{a^2}{\bar{\sigma}^2 + N_\textup{f}}\bar{\sigma}^4 + b^2
\end{aligned}
\end{equation}
and
\begin{equation}
\label{sigmabarstat}
\begin{aligned}
		\bar{\sigma}^2 = \frac{a^2 N_\textup{f}^2}{\bar{\sigma}^2 + N_\textup{f}}\bar{\sigma}^2 + \frac{a^2PN}{(P + N)^2}\sigma^2 
\end{aligned}
\end{equation}
The pair of equations are equivalent to a couple of forth order polynomial equations in the two variables $(\sigma^2, \bar{\sigma}^2)$, and solving these equations can be found easily using standard numerical tools.

\section{Conclusions}
We considered the problem of optimal encoder/decoder filter design over a Shannon Gaussian channel with \textit{noisy feedback} to estimate the state of a scalar linear dynamical system.  
We showed that optimal encoders and decoders are linear filters
with a finite memory and we give explicitly the state space realization of the optimal filters. We also presented the solution of the case where the transmitter has access to noisy measurements of the state. We derived a separation principle for this communication scheme. Necessary and sufficient conditions for the existence of a stationary solution where also given. 

Future work will consider the case where the noise process $n$ is colored for some linear filter $\mathbf{S}\ne I$. Also, the non-scalar case is challenging as we can't rely on the information theoretic inequalities used in this paper for the higher dimensional case.

\bibliography{../../ref/mybib}

\begin{thebibliography}{10}
\providecommand{\url}[1]{#1}
\csname url@samestyle\endcsname
\providecommand{\newblock}{\relax}
\providecommand{\bibinfo}[2]{#2}
\providecommand{\BIBentrySTDinterwordspacing}{\spaceskip=0pt\relax}
\providecommand{\BIBentryALTinterwordstretchfactor}{4}
\providecommand{\BIBentryALTinterwordspacing}{\spaceskip=\fontdimen2\font plus
\BIBentryALTinterwordstretchfactor\fontdimen3\font minus
  \fontdimen4\font\relax}
\providecommand{\BIBforeignlanguage}[2]{{%
\expandafter\ifx\csname l@#1\endcsname\relax
\typeout{** WARNING: IEEEtran.bst: No hyphenation pattern has been}%
\typeout{** loaded for the language `#1'. Using the pattern for}%
\typeout{** the default language instead.}%
\else
\language=\csname l@#1\endcsname
\fi
#2}}
\providecommand{\BIBdecl}{\relax}
\BIBdecl

\bibitem{shannon:48}
C.~E. Shannon, ``A mathematical theory of communication,'' \emph{Bell System
  Tech. J.}, vol.~27, pp. 379--423 and 623--656, 1948.

\bibitem{shannon1949}
------, ``Communication in the presence of noise,'' \emph{Proc. Institute of
  Radio Engineers}, vol.~37, no.~1, pp. 10--21, 1949.

\bibitem{pilc:1969}
R.~J. Pilc, ``The optimum linear modulator for a {G}aussian source with a
  gaussian channel,'' \emph{The Bell System Technical Journal}, pp. 3075--3089,
  November 1969.

\bibitem{kalman:1960}
R.~E. Kalman, ``A new approach to linear filtering and prediction problems,''
  \emph{Trans. of the ASME-Journal of Basic Engineering}, vol.~82, pp. 35--45,
  1960.

\bibitem{bansal:basar:1988}
\BIBentryALTinterwordspacing
R.~Bansal and T.~Basar, \emph{\BIBforeignlanguage{English}{Simultaneous design
  of communication and control strategies for stochastic systems with
  feedback}}, ser. Lecture Notes in Control and Information Sciences,
  A.~Bensoussan and J.~Lions, Eds.\hskip 1em plus 0.5em minus 0.4em\relax
  Springer Berlin Heidelberg, 1988, vol. 111. [Online]. Available:
  \url{http://dx.doi.org/10.1007/BFb0042248}
\BIBentrySTDinterwordspacing

\bibitem{tatikonda:2004}
S.~Tatikonda, A.~Sahai, and S.~Mitter, ``Stochastic linear control over a
  communication channel,'' \emph{IEEE Trans. on Automatic Control}, vol.~49,
  no.~9, pp. 1549--1561, 2004.

\bibitem{minero:2009}
P.~Minero, M.~Franceschetti, S.~Dey, and G.~Nair, ``Data rate theorem for
  stabilization over time-varying feedback channels,'' \emph{Automatic Control,
  IEEE Transactions on}, vol.~54, no.~2, pp. 243--255, Feb 2009.

\bibitem{martins:2008}
N.~Martins and M.~Dahleh, ``Feedback control in the presence of noisy channels:
  Fundamental limitations of performance,'' \emph{Automatic Control, IEEE
  Transactions on}, vol.~53, no.~7, pp. 1604--1615, 2008.

\bibitem{joh10acc}
E.~Johannesson, A.~Rantzer, B.~Bernhardsson, and A.~Ghulchak, ``Encoder and
  decoder design for signal estimation,'' in \emph{American Control
  Conference}, Baltimore, Maryland, USA, June 2010.

\bibitem{kim:2010}
Y.-H. Kim, ``Feedback capacity of stationary {G}aussian channels,''
  \emph{Information Theory, IEEE Transactions on}, vol.~56, no.~1, pp. 57--85,
  Jan 2010.

\bibitem{gattami:ifac:2014}
A.~Gattami, ``Kalman meets {S}hannon,'' in \emph{IFAC World Congress}, August
  2014.

\bibitem{cover:2006}
T.~Cover and J.~A. Thomas, \emph{Elements of Information Theory}.\hskip 1em
  plus 0.5em minus 0.4em\relax John Wiley \& Sons, 2006.

\bibitem{shiryaev}
A.~N. Shiryaev, \emph{Probability}.\hskip 1em plus 0.5em minus 0.4em\relax
  Springer, 1996.

\bibitem{gamal:nit}
A.~E. Gamal and Y.-H. Kim, \emph{\BIBforeignlanguage{English}{Network
  Information Theory}}.\hskip 1em plus 0.5em minus 0.4em\relax Cambridge
  University Press, 2012.

\bibitem{gallager}
R.~G. Gallager, \emph{\BIBforeignlanguage{English}{Information theory and
  reliable communication}}.\hskip 1em plus 0.5em minus 0.4em\relax Wiley, New
  York, 1968.

\bibitem{astrom:1970}
K.~J. {\AA}str{\"o}m, \emph{Stochastic Control Theory}.\hskip 1em plus 0.5em
  minus 0.4em\relax Academic Press, 1970.

\end{thebibliography}


\section*{Appendix}
 

\subsection*{Proof of Theorem \ref{mainy}}

	Suppose that $\E \{g_t(x^t)\} = \alpha_t$ where $\{\alpha_k\}_{k=0}^t$ are
	deterministic real numbers independent of $x^t$ and are known at the encoder $g_t$ 
	and decoder $f_t$. 
 Note that $y(t) = g_t(x^t)+n(t)$. The estimate of $x(t+1)$ based on $y(k)$, 
 $k=0, ..., t$, is the same as the estimate of $x(t+1)$ based on $y(k)-\alpha_k$ for $k=0, ..., t$ since $\alpha_k$ is deterministic and known at the decoder.
But it means that we can replace $g_t(x^t)$ with $g_t'(x^t)=g(x^t)-\alpha_t$, and $g_t'(x^t)$ satisfies both $\E \{g_t'(x^t)\} = 0$ and the power constraint $\E|g_t'(x^t)|^2 \leq P$ since
\begin{equation*}
 \begin{aligned}
\E|g_t'(x^t)|^2 	&= \E|g_t(x^t)-\alpha|^2 \\
			&= \E |g_t(x^t)|^2 - \alpha^2\\ 
			&= P-\alpha^2\leq P.
 \end{aligned}
 \end{equation*} 
 Thus, without loss of generality,  we may restrict the encoders $g$ to the set 
$$\{ g~ | ~\E \{g(x^t)\} = 0\}.$$ 

We will now prove that the optimal filters are linear by induction. 
Suppose that $g_{k-1}$ and $f_{k-1}$ are linear for $k=1, ..., t$. 
Then, $\tilde{x}^t$, $x^t$, $y^{t-1}$, and $y_{\textup{f}}^{t-1}$ are jointly Gaussian. 

Let $\hat{x}(t|t) = f'_t(y^t)$ be the optimal estimate of $x(t)$ based on $y^t$ and  
let $\tilde{x}(t|t)  = x(t) - \hat{x}(t|t)$, for $t=0, ..., T$. Then, $f'_t(y^t) = \E\{x(t) | y^t\}$ according to Proposition \ref{lse}. Now we have that
\begin{equation}
 \label{nofbestimate1}
 	\begin{aligned}
		 \hat{x}(t|t) 	&= \E\{x(t)|y^t\}\\
		 				&=  \E\{(\hat{x}(t) + \tilde{x}(t)  | y^{t}\}\\
						&= \hat{x}(t) + \E\{\tilde{x}(t)  | y^t \},
 	\end{aligned}
 \end{equation} 
\begin{equation}
 \label{nofbestimate2}
 	\begin{aligned}
\tilde{x}(t+1) 	&= x(t+1) -\hat{x}(t+1) \\
				&= ax(t)+bw(t)-a\hat{x}(t|t)\\ &= a\tilde{x}(t|t) + bw(t)\\
\end{aligned}
 \end{equation} 
We see that minimizing $\E |\tilde{x}(t+1)|^2$ is equivalent to minimizing the mean square error of
$$
\tilde{x}(t|t) = \tilde{x}(t) - \E\{\tilde{x}(t)  | y^t \}
$$
at the decoder. Now introduce

$$
\check{x}(t) := \E\{\tilde{x}(t)  | x^t, y_\textup{f}^{t-1} \}
$$
and
$$
\bar{x}(t) := \tilde{x}(t) - \check{x}(t) .
$$
Then, $\check{x}(t)$ is a linear function of $x^t$ and $y_\textup{f}^{t-1}$, since $\tilde{x}(t)$,  $x^t$, and $y_\textup{f}^{t-1}$ are jointly Gaussian by the induction hypothesis.
Thus, $\bar{x}(t)$ is independent of $\check{x}(t)$, $x^t$, and $y_\textup{f}^{t-1}$. This implies that $\bar{x}(t)$ is independent of $g_t(x^t, y_\textup{f}^{t-1})$ and $g_t(x^t, y_\textup{f}^{t-1}) + n(t) = y(t)$.

The Markov chain 
$$\check{x}(t) \rightarrow g_t(x^t, y_\textup{f}^{t-1}) \rightarrow y(t) 
= g_t(x^t, y_\textup{f}^{t-1}) + n(t),$$
together with Proposition \ref{dataproc}, gives
 
 \begin{equation}
 \label{nofbestimate4}
	\begin{aligned}
 		I(\check{x}(t); y(t))	&\leq I(g_t(x^t, y_\textup{f}^{t-1}); y(t)). 
	\end{aligned}	
 \end{equation} 
 
 The Shannon capacity of a Gaussian channel gives an upper bound for the mutual information between the transmitted message $z(t)=g_t({x}^{t}, y_\textup{f}^{t-1})$ and received message $y(t)$ (see \cite{gallager}):
 \begin{equation}
 \label{nofbestimate5}
 I(g_t(x^t, y_\textup{f}^{t-1}); y(t))	\leq 	\frac{1}{2} \log_2{\left(1+\frac{P}{N}\right)}.
 \end{equation}
 Combining (\ref{nofbestimate4})-(\ref{nofbestimate5}), we get
\begin{equation}
\label{nofbestimate6}
	2^{-2I(\check{x}(t); y(t))} \geq \frac{N}{P+N}
\end{equation}
with equality if $\check{x}(t)$ and $y(t)$ are mutually Gaussian and 
$g_t(x^t, y_\textup{f}^{t-1})=\frac{\sqrt{P}}{\sigma_t} \check{x}(t)$ 
with $\sigma_t^2=  \E |\check{x}(t)|^2$. 
From the definition of mutual information, we have that
\begin{equation}
\label{nofbestimate7}
 	\begin{aligned}
		h(\check{x}(t) |y(t))	&= h(\check{x}(t) ) - I(\check{x}(t) ; y(t)).\\
 	\end{aligned}
 \end{equation} 
Now we get
\begin{eqnarray}
		2\pi e \E\{|\tilde{x}(t|t)|^2\}	&\geq& 2^{2h(\tilde{x}(t) | y^t)} \label{eq1} \\
							&=& 2^{2h(\tilde{x}(t) | y(t) )} \label{eq2}\\
							&=& 2^{2h(\check{x}(t) + \bar{x}(t) | y(t))} \nonumber\\
							&=& 2^{2h(\check{x}(t) | y(t))+ h(\bar{x}(t))} \label{eq3}\\
							&\geq& 2^{2h(\check{x}(t) | y(t))} + 2^{2h(\bar{x}(t))} \label{eq4}\\
							&=& 2^{2h(\check{x}(t) ) - 2I(\check{x}(t) ; y(t))} + 2^{2h( \bar{x}(t) )} ~~~~ \label{eq5}\\
							&\geq& \frac{N}{P+N} 2^{2h(\check{x}(t) )} + 2^{2h( \bar{x}(t) )}\label{eq6}
 \end{eqnarray} 
where (\ref{eq1}) follows from Proposition \ref{logerror}(with equality if $\tilde{x}(t)$ and $y^t$ are jointly Gaussian), (\ref{eq2}) follows from the fact that
 $\tilde{x}(t)$ is independent of $y^{t-1}$, (\ref{eq3}) follows from the fact that $\bar{x}(t)$ is independent of $\check{x}(t)$ and $y(t)$,  (\ref{eq4}) follows from the entropy power inequality(Proposition \ref{entpineq}), (\ref{eq5}) follows from equation (\ref{nofbestimate7}), and (\ref{eq6}) follows from inequality (\ref{nofbestimate6}). Furthermore, equality holds in (\ref{eq1})-(\ref{eq6}) if 
 $$z(t) = g_t(x^t, y_\textup{f}^{t-1}) = \frac{\sqrt{P}}{\sigma_t} \check{x}(t)$$
 with $\sigma_t^2=  \E |\check{x}(t)|^2$. This completes the proof.

\subsection*{Proof of Theorem \ref{mainnofb2}}
Let
$
\hat{x}(t) = \E\{ x(t)  | y^{t-1} \}
$
,
$
\tilde{x}(t) = x(t) - \hat{x}(t), 
$
$
\hat{x}(t|t) = \E\{ x(t)  | y^t \},
$
and
$
\tilde{x}(t|t) = x(t) - \hat{x}(t|t).
$

Then,
\begin{equation}
 \label{xhat}
 	\begin{aligned}
\hat{x}(t+1) 	&= a\hat{x}(t|t)\\
				&= a \E\{\hat{x}(t) + \tilde{x}(t)  | y^t \}\\
				&= a \hat{x}(t) + a \E\{\tilde{x}(t)  | y(t) \}\\
\end{aligned}
 \end{equation} 
and
\begin{equation}
 \label{tildex}
 	\begin{aligned}
\tilde{x}(t+1) 	
				&= a\tilde{x}(t) - a\E\{\tilde{x}(t)  | y(t) \} + bw(t).\\
\end{aligned}
 \end{equation} 
 
According to Theorem  \ref{mainy}, the optimal signal $z$ is given by 
\begin{equation*}
 	\begin{aligned}
		\check{x}(t)  	&=  \E\{\tilde{x}(t) | x^{t}, y_\textup{f}^{t-1}\}\\
 		z(t) 			&= \frac{\sqrt{P}}{\sigma_t} \check{x}(t) 
	\end{aligned}
 \end{equation*} 
 with $\sigma_t^2=  \E|\check{x}(t)|^2$. 
 Now recall that $y(t) = z(t) + n(t)$, 
 $
 	\check{x}(t) =  \E\{\tilde{x}(t) | x^{t}, y_\textup{f}^{t-1} \}
 $
,
 $
 \bar{x}(t)  = \tilde{x}(t)  - \check{x}(t) 
 $
  , and $ \bar{x}(t) $ is orthogonal to $x^t$ and hence to $y(t)$. Since
 $\tilde{x}(t)$ and $y(t)$ are jointly Gaussian, 
 $\E\{\tilde{x}(t) | y(t)\}$ is a linear function of $y(t)$ given by
 \begin{equation}
 \label{tildex_estimate}
 	\begin{aligned}
 		\E\{\tilde{x}(t) | y(t)\} 	&= \E\{\check{x}(t) + \bar{x}(t) | y(t)\}\\
						&= \E\{\check{x}(t) | y(t)\}+ \E\{\bar{x}(t) | y(t)\}\\
						&= \E\{\check{x}(t) | y(t)\}\\
						&= \mathbf{cov}\{\check{x}(t),y(t)\} (\mathbf{cov}\{y(t),y(t)\})^{-1}y(t)\\
						&= \kappa(t) y(t)
 	\end{aligned}
 \end{equation} 
 with 
 \begin{equation}
 \label{kappa}
 	\kappa(t) = \sigma_t\sqrt{P}(P+N)^{-1}.
\end{equation}
Then, (\ref{xhat})-(\ref{tildex_estimate}) imply
\begin{equation}
 \label{xhat2}
 	\begin{aligned}
		\hat{x}(t+1) 	&= a \hat{x}(t) + a \kappa(t) y(t)\\
					&= a \hat{x}(t) + a \kappa(t) \frac{\sqrt{P}}{\sigma_t} \check{x}(t) 
							+ a\kappa(t) n(t),
	\end{aligned}
 \end{equation} 

\begin{equation}
\label{xtilde}
 	\begin{aligned}
		\tilde{x}(t+1) 	&= a\tilde{x}(t) - a\kappa(t) y(t) + bw(t)\\
						&= a\tilde{x}(t) - a\kappa(t) \frac{\sqrt{P}}{\sigma_t} \check{x}(t) - 
							a\kappa(t) n(t) + bw(t)\\
\end{aligned}
 \end{equation} 
 
 The encoder has access to $ \check{x}^t$ at time $t+1$. It has also access to
$z^t$ and $y_\textup{f}^t$, which implies that it has access to
$$
y_\textup{f}(k) - z(k) = n(k) + n_\textup{f}(k)
$$ 
for $k = 1, ..., t$. Now we have that
\begin{equation}
\label{nhat}
	\begin{aligned}
		\hat{n}(t) 	&= \E\{n(t) | n(t) + n_\textup{f}(t)\} \\
					&= \frac{N}{N + N_\textup{f}} (n(t) + n_\textup{f}(t))
	\end{aligned}
 \end{equation} 
 and
\begin{equation}
\label{ntilde}
	\begin{aligned}
		\tilde{n}(t) 	&= n(t) -\hat{n}(t) \\
					&= \frac{N_\textup{f}}{N + N_\textup{f}} n(t)  - \frac{N}{N + N_\textup{f}}n_\textup{f}(t)
	\end{aligned}
 \end{equation}

We will show that
\begin{equation}
\label{checkx}
	\check{x}(t+1) = a\check{x}(t) - a\kappa(t) \frac{\sqrt{P}}{\sigma_t} \check{x}(t) 
							 - a\kappa(t)\hat{n}(t) + bw(t)\\
\end{equation}
and

\begin{equation}
\label{barx}
	\bar{x}(t+1) = a\bar{x}(t) 
							 - a\kappa(t)\tilde{n}(t) \\
\end{equation}

First we note that $\bar{x}$ as defined in (\ref{barx}) depends only on the channel noise
estimation error $\tilde{n}$ and is therefore independent of $x$, $y_\textup{f}$, and $\check{x}$. Now (\ref{checkx})-(\ref{barx}) give

\begin{equation}
\label{erreq}
 	\begin{aligned}
		 	\check{x}(t+1) &+ \bar{x}(t+1)\\
						&= a\check{x}(t) - a\kappa(t) \frac{\sqrt{P}}{\sigma_t} \check{x}(t) 
							 - a\kappa(t)\hat{n}(t) + bw(t)\\ 
							&~~~+ a\bar{x}(t) 
							 - a\kappa(t)\tilde{n}(t)\\
						&= a(\check{x}(t) + \bar{x}(t)) - 
							a\kappa(t) \frac{\sqrt{P}}{\sigma_t} \check{x}(t) \\
							&~~~ - a\kappa(t) (\hat{n}(t) + \tilde{n}(t))+ bw(t)\\	
						&= a\tilde{x}(t) - a\kappa(t) \frac{\sqrt{P}}{\sigma_t} \check{x}(t)
						- a\kappa(t) n(t)+ bw(t),\\
	\end{aligned}
 \end{equation} 
 which is exactly the expression for the dynamics of $\tilde{x}(t+1)$ 
given by (\ref{xtilde}). This establishes (\ref{checkx}) - (\ref{barx}). Now we have
\begin{equation}
\label{x-s}
 	\begin{aligned}
		x(t)	&= \E\{x(t) | x^t, y_\textup{f}^{t-1}\}\\
			&= \E\{\hat{x}(t) + \tilde{x}(t) | x^t, y_\textup{f}^{t-1}\}\\
			&= \E\{\hat{x}(t)  | x^t, y_\textup{f}^{t-1}\} + \E\{ \tilde{x}(t) | x^t, y_\textup{f}^{t-1}\}\\
			&= \E\{\hat{x}(t)  | x^t, y_\textup{f}^{t-1}\} + \check{x}(t).
	\end{aligned}
\end{equation} 
From equation (\ref{xhat2}), we see that 
\begin{equation}
 \label{ss_checkx1}
	\E\{\hat{x}(t)  | x^t, y_\textup{f}^{t-1}\} = s(t)
\end{equation} 
where
\begin{equation}
 	\label{ss_checkx2}
 	\begin{aligned}
		s(t+1) 	&= a s(t) + a \kappa(t) \frac{\sqrt{P}}{\sigma_t} \check{x}(t) + a\kappa(t) \hat{n}(t)\\
				&= a s(t) + K(t)( z(t) + \hat{n}(t))
	\end{aligned}
\end{equation} 
 since the noise signal $\tilde{n}$ is independent of $x$ and $y_\textup{f}$.
Finally, combining (\ref{x-s}) - (\ref{ss_checkx2}) gives 
\begin{equation}
\label{x-s2}
	 \check{x}(t) = x(t) - s(t).
\end{equation}
Now set 
\begin{equation}
\label{K}
K(t) = a\kappa(t). 
\end{equation}
Then,
\begin{equation}
\label{ss_checkx3}
\begin{aligned}
s(t+1) 	&= a s(t) + a \kappa(t) \frac{\sqrt{P}}{\sigma_t} \check{x}(t) + a\kappa(t) \hat{n}(t) \\
		&= a s(t) + a \frac{P}{P+N} \check{x}(t) + K(t) \hat{n}(t) \\
		&= \left(a - \frac{aP}{P+N}\right) s(t) + \frac{aP}{P+N} {x}(t) + K(t) \hat{n}(t)\\
		&= \frac{aN}{P+N} s(t) +  \frac{aP}{P+N} {x}(t) + K(t) \hat{n}(t)\\
\end{aligned}
\end{equation}
and
\begin{equation*}
\begin{aligned}
\begin{bmatrix}
	s(t+1)\\
	x(t+1)	
\end{bmatrix}
&=
\begin{bmatrix}
	\frac{aN}{P+N} 	& \frac{aP}{P+N}\\
	0		& a	
\end{bmatrix} 
\begin{bmatrix}
	s(t)\\
	x(t)	
\end{bmatrix}+
\begin{bmatrix}
	K(t)   	& 0\\
	0 		& b	
\end{bmatrix}
\begin{bmatrix}
	 \hat{n}(t)\\
	w(t)	
\end{bmatrix}
\end{aligned}
\end{equation*}
Introduce the covariance matrix
$$
\begin{bmatrix}
	V_{ss}(t) & V_{sx}(t)\\
	V_{xs}(t) & V_{xx}(t)	
\end{bmatrix}
=
\E 
\begin{bmatrix}
	s(t)\\
	x(t)	
\end{bmatrix}
\begin{bmatrix}
	s(t)\\
	x(t)	
\end{bmatrix}^\intercal.
$$
Since $\hat{n}(t)$ and $w(t)$ are uncorrelated with $x(t)$ and $s(t)$, we get 
\begin{equation}
\label{V}
\begin{aligned}
&
\begin{bmatrix}
	V_{ss}(t+1) & V_{sx}(t+1)\\
	V_{xs}(t+1) & V_{xx}(t+1)	
\end{bmatrix} \\
&=
\E 
\begin{bmatrix}
	s(t+1)\\
	x(t+1)	
\end{bmatrix}
\begin{bmatrix}
	s(t+1)\\
	x(t+1)	
\end{bmatrix}^\intercal\\
& =
\E \left\{\left(
\begin{bmatrix}
	\frac{aN}{P+N} 	& \frac{aP}{P+N}\\
	0		& a	
\end{bmatrix} 
\begin{bmatrix}
	s(t)\\
	x(t)	
\end{bmatrix}+
\begin{bmatrix}
	K(t)   	& 0\\
	0 		& b	
\end{bmatrix}
\begin{bmatrix}
	 \hat{n}(t)\\
	w(t)	
\end{bmatrix}
\right)\times \right. \\
&
\left.
\hspace{9.5mm}
\left(
\begin{bmatrix}
	\frac{aN}{P+N} 	& \frac{aP}{P+N}\\
	0		& a	
\end{bmatrix} 
\begin{bmatrix}
	s(t)\\
	x(t)	
\end{bmatrix}+
\begin{bmatrix}
	K(t)   	& 0\\
	0 		& b	
\end{bmatrix}
\begin{bmatrix}
	 \hat{n}(t)\\
	w(t)	
\end{bmatrix}
\right)^\intercal
\right\} \\
&=
\begin{bmatrix}
	\frac{aN}{P+N} 	& \frac{aP}{P+N}\\
	0		& a	
\end{bmatrix} 
\begin{bmatrix}
	V_{ss}(t) & V_{sx}(t)\\
	V_{xs}(t) & V_{xx}(t)	
\end{bmatrix}
\begin{bmatrix}
	\frac{aN}{P+N} 	& \frac{aP}{P+N}\\
	0		& a	
\end{bmatrix}^\intercal \\
&
\hspace{4mm}
+ 
\begin{bmatrix}
	\frac{K^2(t)N^2}{N+N_\textup{f}} & 0\\
	0 & b^2	
\end{bmatrix}.
\end{aligned}
\end{equation}
Thus, 

\begin{equation}
\label{sigmat}
\begin{aligned}
\sigma_t^2 &=  \E |\check{x}(t)|^2 \\
			&= \E|x(t)-s(t)|^2\\ 
			&= V_{xx}(t) -2V_{sx}(t)+ V_{xx}(t).
\end{aligned}
\end{equation}
Putting together (\ref{kappa}), (\ref{xhat2}), and (\ref{ss_checkx2}) -
(\ref{sigmat}) gives the desired result.

\subsection*{Proof of Theorem \ref{mainfb2}}

Define the estimate $\breve{x}(t|t-1) = \E\{x(t) | \gamma^{t-1}\}$ and let
$$\xi(t) = x(t) - \breve{x}(t|t-1)$$ 
be the estimation error. It's well known that $\breve{x}(t)$ is given by the Kalman filter
\begin{equation} 
\label{transs}
\begin{aligned}
	\breve{x}(t)		&= \breve{x}(t|t-1) + L(t)(c\xi(t)+dv(t))\\
	\breve{x}(t+1|t)	&= a\breve{x}(t)\\
						&= a\breve{x}(t|t-1) + aL(t)(c\xi(t)+dv(t))\\
	\xi(t+1) 			&= (a-aL(t)c)\xi(t) + bw(t) - aL(t)dv(t) 
\end{aligned}
\end{equation}
where $L(t)$ are the optimal Kalman filter gains for $t=0,..., T$ (see, e. g., \cite{astrom:1970}): 
\begin{equation*} 
\begin{aligned}
L(t)					&= V_{\xi \xi}(t)c(c^2V_{\xi \xi}(t) + d^2V_{vv}(t) )^{-1} \\ 
V_{\xi \xi}(t+1) 	&=  (a-aL(t)c)^2 V_{\xi \xi}(t) \\
					&+
\begin{bmatrix}
	b & -aL(t)
\end{bmatrix}
\begin{bmatrix}
	V_{ww}(t) & V_{wv}(t)\\
	V_{vw}(t) & V_{vv}(t)	
\end{bmatrix}
\begin{bmatrix}
	b & -aL(t)
\end{bmatrix}^\intercal
\end{aligned}
\end{equation*}

We also know that $\gamma^{t-1}$ and $\xi(t)$ are uncorrelated according to Proposition \ref{lse}. This implies in turn that $y^{t-1}$ and $\xi(t)$ are uncorrelated.
Hence, the averaged estimation error of the decoder is equal to

$$
\frac{1}{T}\sum_{t=1}^T \E|x(t)-\hat{x}(t)|^2 = 
\frac{1}{T}\sum_{t=1}^T \left(\E|\breve{x}(t)-\hat{x}(t)|^2 + \E|\xi(t)|^2\right).
$$

Obviously, the decoder can't do much about the error covariance $\E|\xi(t)|^2$.
The decoder $\F$ minimizes the averaged estimation error above if and only if it minimizes the
averaged estimation error of $\breve{x}(t)$. 
Thus, we have transformed the output measurement problem to a state measurement 
problem at the encoder $\G$, where the measured state is the state $\breve{x}(t)$ 
of the linear time-varying dynamical system given by 
\begin{equation*} 
\label{transs2}
	\begin{aligned}
	\breve{x}(t+1)		&= \breve{x}(t+1|t) + L(t+1)(c\xi(t+1)+dv(t+1))\\
						&= a\breve{x}(t) + L(t+1)(c\xi(t+1)+dv(t+1))\\
						&= a\breve{x}(t) + \beta(t) \omega(t)
	\end{aligned}
\end{equation*}
with $\omega(t) \sim \mathcal{N}(0,1)$ and
\begin{equation*} 
	\begin{aligned}
	\beta^2(t) &= L^2(t+1)\E\{(c\xi(t+1)+dv(t+1))^2\} \\
				&= L^2(t+1)  (c^2 V_{\xi \xi}(t+1)  + d^2 V_{vv}(t+1) )
	\end{aligned}
\end{equation*}
Inserting $b(t) = \beta(t)$ in Problem \ref{p1} and using
Theorem \ref{mainnofb2} gives the (\ref{K2})-(\ref{cov2}). 
This concludes the proof.

\subsection*{Proof of Theorem \ref{mainxhat2}}
Similar to Theorem \ref{mainnofb2}, we have that
\begin{equation}
 	\begin{aligned}
		\hat{x}(t+1) 	&= a \hat{x}(t) + K(t) \frac{\sqrt{P}}{\sigma_t} \check{x}(t) 
							+ K(t) n(t),
	\end{aligned}
 \end{equation} 

\begin{equation}
 	\begin{aligned}
		\tilde{x}(t+1) 	&= a\tilde{x}(t) - K(t) \frac{\sqrt{P}}{\sigma_t} \check{x}(t) - 
							K(t) n(t) + bw(t)\\
\end{aligned}
 \end{equation} 
with, as before, 
\begin{eqnarray}
	&K(t) &= a\sigma_t \sqrt{P}(P+N)^{-1}\nonumber \\
	&\sigma_t^2 &= \E \check{x}^2(t) 
\end{eqnarray}
Clearly, 
\begin{equation}
\label{xcheckxhat}
 	\begin{aligned}
\check{x}(t) 	&= \E\{\tilde{x}(t)  | x^t, y_\textup{f}^{t-1}\} \\
				&= \E\{\tilde{x}(t)  | x^{t-1}, w(t-1), y_\textup{f}^{t-1}\}\\
				&=  \E\{a\tilde{x}(t-1)  | x^{t-1}, w(t-1), y_\textup{f}^{t-1}\} \\
				&~~~ - K(t-1) \frac{\sqrt{P}}{\sigma_{t-1}} \check{x}(t-1) + bw(t-1)\\ 
				&=a\check{x}(t-1) + \E\{a\bar{x}(t-1)  | x^{t-1}, w(t-1), y_\textup{f}^{t-1}\} \\
				&~~~ - K(t-1) \frac{\sqrt{P}}{\sigma_{t-1}} \check{x}(t-1) + bw(t-1)\\
\end{aligned}
 \end{equation} 
The transmitter can consctruct the new measurement 
$$x(t-1) - \check{x}(t-1) - y_\textup{f}(t-1) = \bar{x}(t-1) + n_\textup{f}(t-1),$$
so
\begin{equation*}
 	\begin{aligned}
&\E \{ a\bar{x}(t-1)  | x^{t-1}, w(t-1), y_\textup{f}^{t-1}\} \\
&=\E\{a\bar{x}(t-1)  | x^{t-1}, w(t-1), y_\textup{f}^{t-2}, \bar{x}(t-1) + n_\textup{f}(t-1)\}
\end{aligned}
 \end{equation*}
Since $\bar{x}(t-1)$ and $n_\textup{f}(t-1)$ are independent of $x^{t-1}$, $w(t-1)$, and $y_\textup{f}^{t-2}$, we have that 

\begin{equation*}
 	\begin{aligned}
&\E \{ a\bar{x}(t-1)  | x^{t-1}, w(t-1), y_\textup{f}^{t-1}\} \\
&=\E \{ a\bar{x}(t-1)  | \bar{x}(t-1) + n_\textup{f}(t-1)\}
\end{aligned}
 \end{equation*}
 Let $\bar{\sigma}_t^2 = \E \bar{x}^2(t) =  \E (\tilde{x}(t) - \check{x}(t))^2 = \E \tilde{x}^2(t) -  \E \check{x}^2(t)$. Then,
\begin{equation}
\label{xbarestimate}
 	\begin{aligned}
\E \{ & a\bar{x}(t)  | \bar{x}(t) + n_\textup{f}(t)\} \\
&= a\bar{\sigma}_t^2 (\bar{\sigma}_t^2 + N_\textup{f})^{-1}(\bar{x}(t) + n_\textup{f}(t))
\end{aligned}
\end{equation}
Equation (\ref{xbarestimate}) together with (\ref{xcheckxhat}) gives
\begin{equation}
\label{xcheckxhat2}
 	\begin{aligned}
\check{x}(t+1)
				&= a\check{x}(t) - K(t) \frac{\sqrt{P}}{\sigma_{t}} \check{x}(t) + bw(t)  \\
				& ~~~+ a\bar{\sigma}_t^2 (\bar{\sigma}_t^2 + N_\textup{f})^{-1}(\bar{x}(t) + n_\textup{f}(t))\\
				&= aN(P+N)^{-1} \check{x}(t) + bw(t)  \\
				& ~~~+ a\bar{\sigma}_t^2 (\bar{\sigma}_t^2 + N_\textup{f})^{-1}(\bar{x}(t) + n_\textup{f}(t))\\
\end{aligned}
 \end{equation} 
and we can verify that
\begin{equation}
\label{xbarxhat}
 	\begin{aligned}
\bar{x}(t+1)
				&= a\bar{x}(t)  - K(t)n(t) \\
				&~~~- a\bar{\sigma}_t^2 (\bar{\sigma}_t^2 + N_\textup{f})^{-1}(\bar{x}(t) + n_\textup{f}(t))\\
				&= aN_\textup{f} (\bar{\sigma}_t^2 + N_\textup{f})^{-1}\bar{x}(t)  - K(t)n(t) \\
				&~~~ -a\bar{\sigma}_t^2 (\bar{\sigma}_t^2 + N_\textup{f})^{-1} n_\textup{f}(t)\\
\end{aligned}
 \end{equation} 
Now by substituting $bw(t) = x(t+1)-ax(t)$ and $\bar{x}(t) + n_\textup{f}(t)= x(t) - \check{x}(t) - y_\textup{f}(t)$ in (\ref{xcheckxhat2}), we
get
\begin{equation}
\label{xcheckxhat3}
 	\begin{aligned}
\check{x}(t+1)
			&= aN(P+N)^{-1} \check{x}(t) + x(t+1)-ax(t)  \\
			& ~~~+ a\bar{\sigma}_t^2 (\bar{\sigma}_t^2 + N_\textup{f})^{-1}(x(t) - \check{x}(t) - y_\textup{f}(t))\\
\end{aligned}
 \end{equation}

The dynamical system equations given by (\ref{xcheckxhat2}) and (\ref{xbarxhat}) give the dynamics of the variance values $\sigma_t^2$ and $\bar{\sigma}_t^2$
\begin{equation}
\begin{aligned}
		\sigma_{t+1}^2 = \frac{a^2 N^2}{(P + N)^{2}}\sigma_t^2 + \frac{a^2}{\bar{\sigma}_t^2 + N_\textup{f}}\bar{\sigma}_t^4 + b^2
\end{aligned}
\end{equation}
and
\begin{equation}
\begin{aligned}
		\bar{\sigma}_{t+1 }^2 = \frac{a^2 N_\textup{f}^2}{\bar{\sigma}_t^2 + N_\textup{f}}\bar{\sigma}_t^2 + \frac{a^2PN}{(P + N)^2}\sigma_t^2 
\end{aligned}
\end{equation}

\end{document}